\documentclass[11pt]{article}
\usepackage{amsmath, amsthm, amssymb, dsfont, mathrsfs}
\usepackage{fullpage}
\usepackage{times}
\usepackage{float}
\usepackage{graphicx, color}

\floatstyle{boxed}
\newfloat{constr}{htb!}{lop}
\floatname{constr}{Construction}

\renewcommand{\le}{\leqslant}
\renewcommand{\leq}{\leqslant}
\renewcommand{\ge}{\geqslant}
\renewcommand{\geq}{\geqslant}
\title{Non-Malleable Coding Against Bit-wise and Split-State Tampering\thanks{
A preliminary version of this article appears under the same title in 
proceedings of Theory of Cryptography Conference (TCC 2014).
}}

\author{{\sc Mahdi Cheraghchi}\thanks{%
Email: $\langle$mahdi@csail.mit.edu$\rangle$.
Research supported in part by V. Guruswami's Packard Fellowship, MSR-CMU Center for Computational Thinking,
and the Swiss National Science Foundation research grant PA00P2-141980.} \\%
 CSAIL \\
  MIT\\
  Cambridge, MA 02139
 \and %
{\sc Venkatesan Guruswami}\thanks{
Email: $\langle$guruswami@cmu.edu$\rangle$.  Research supported in part by the National Science Foundation under Grant No. CCF-0963975.  Any opinions, findings, and conclusions or recommendations expressed in this material are those of the author(s) and do not necessarily reflect the views of the National Science Foundation.} %
 \\
  Computer Science Department \\
  Carnegie Mellon University \\
  Pittsburgh, PA 15213
}

\date{}

\newcommand{\cC}{\mathcal{C}}
\newcommand{\cD}{\mathcal{D}}
\newcommand{\cE}{\mathcal{E}}
\newcommand{\cX}{\mathcal{X}}
\newcommand{\cY}{\mathcal{Y}}

\newcommand{\cS}{\mathcal{S}}
\newcommand{\cN}{\mathcal{N}}

\newcommand{\cU}{\mathcal{U}}
\newcommand{\U}{\mathcal{U}}
\newcommand{\F}{\mathds{F}}
\newcommand{\cF}{\mathcal{F}}

\newcommand{\E}{\mathds{E}}
\newcommand{\supp}{\mathsf{supp}}

\newcommand{\nm}{\mathsf{NMExt}}
\newcommand{\eps}{\epsilon}

\newtheorem{thm}{Theorem}[section]
\newtheorem{prop}[thm]{Proposition}
\newtheorem{claim}[thm]{Claim}
\newtheorem{lem}[thm]{Lemma}

\newtheorem{coro}[thm]{Corollary}
\theoremstyle{definition}

\newtheorem{defn}[thm]{Definition}
\newtheorem{remark}[thm]{Remark}

\newtheorem*{caveat*}{Caveat}

\newcommand{\dist}{\mathsf{dist}}

\newcommand{\zo}{\{0,1\}}

\newcommand{\poly}{\mathsf{poly}}

\newcommand{\enc}{{\mathsf{Enc}}}
\newcommand{\dec}{{\mathsf{Dec}}}
\newcommand{\same}{{\underline{\mathsf{same}}}}

\newcommand{\Copy}{\mathsf{copy}}
\newcommand{\distr}{\mathscr{D}}
\newcommand{\Perm}{\mathsf{Perm}}

\newcommand{\mchOK}[1]{}
\newcommand{\vnoteOK}[1]{}

\newcommand{\calF}{\mathcal{F}}
\usepackage[
     hidelinks,
     pdftitle={Non-Malleable Coding Against Bit-wise and Split-State Tampering},
     pdfsubject={Non-Malleable Coding Against Bit-wise and Split-State Tampering},
     pdfauthor={M. Cheraghchi and V. Guruswami}]{hyperref}

\begin{document}

\maketitle
\thispagestyle{empty}

\begin{abstract}
Non-malleable coding, introduced by Dziembowski, Pietrzak and Wichs (ICS~2010), 
aims for protecting the integrity of information against tampering attacks in
situations where error-detection is impossible. Intuitively, information 
encoded by a non-malleable code either decodes to the original message or, 
in presence of any tampering, to an unrelated message. Non-malleable coding
is possible against any class of adversaries of bounded size. In particular,
Dziembowski et al.~show that such codes exist and may achieve positive rates 
for any class of tampering functions of size at most $2^{2^{\alpha n}}$, 
for any constant $\alpha \in [0, 1)$. However, this result is existential
and has thus attracted a great deal of subsequent research on explicit constructions of non-malleable codes against natural classes of adversaries.

\smallskip
In this work, we consider constructions of coding schemes against
two well-studied classes of tampering functions; namely, bit-wise tampering
functions (where the adversary tampers each bit of the encoding independently)
and the much more general class of split-state adversaries (where two independent adversaries arbitrarily
tamper each half of the encoded sequence). We obtain the following results for
these models.
\begin{enumerate}
\item For bit-tampering adversaries, we obtain explicit and efficiently
encodable and decodable non-malleable codes of length $n$ achieving rate $1-o(1)$ 
and error (also known as ``exact security'') 
$\exp(-\tilde{\Omega}(n^{1/7}))$. Alternatively, it is possible to
improve the error to $\exp(-\tilde{\Omega}(n))$ at the
cost of making the construction Monte Carlo with success
probability $1-\exp(-\Omega(n))$ (while still allowing a
compact description of the code). Previously, the best known construction
of bit-tampering coding schemes was due to Dziembowski et al.~(ICS~2010),
which is a Monte Carlo construction achieving rate close to $.1887$.

\item We initiate the study of \emph{seedless non-malleable extractors}
as a natural variation of the notion of non-malleable extractors introduced
by Dodis and Wichs (STOC~2009). We show that construction of non-malleable
codes for the split-state model reduces to construction of non-malleable
two-source extractors. We prove a general result on existence of
seedless non-malleable extractors, which implies that codes obtained
from our reduction can achieve rates arbitrarily close to $1/5$ and
exponentially small error.
In a separate recent work, the authors show that
the optimal rate in this model is $1/2$.
Currently, the best known explicit construction of split-state coding
schemes is due to Aggarwal, Dodis and Lovett (ECCC~TR13-081)
which only achieves vanishing (polynomially small) rate.
\end{enumerate}
\end{abstract}
\newpage
\tableofcontents
\thispagestyle{empty}
\newpage

\section{Introduction}
\label{sec:intro}

Non-malleable codes were introduced by Dziembowski, Pietrzak, and
Wichs~\cite{ref:nmc} as a relaxation of the classical notions of
error-detection and error-correction.  Informally, a code is
non-malleable if the decoding a corrupted codeword either recovers the
original message, or a completely unrelated message. Non-malleable
coding is a natural concept that addresses the basic question of
storing messages securely on devices that may be subject to tampering,
and they provide an elegant solution to the problem of protecting the
integrity of data and the functionalities implemented on them against
``tampering attacks"~\cite{ref:nmc}. This is part of a general recent
trend in theoretical cryptography to design cryptographic schemes that
guarantee security even if implemented on
devices that may be subject to physical tampering.  The notion of
non-malleable coding is inspired by the influential theme of
non-malleable encryption in cryptography which guarantees the
intractability of tampering the ciphertext of a message into the
ciphertext encoding a related message.

The definition of non-malleable codes captures the requirement that if
some adversary (with full knowledge of the code) tampers the codeword
$\enc(s)$ encoding a message $s$, corrupting it to $f(\enc(s))$, he
cannot control the relationship between $s$ and the message the
corrupted codeword $f(\enc(s))$ encodes. For this definition to be
feasible, we have to restrict the allowed tampering functions $f$
(otherwise, the tampering function can decode the codeword to compute
the original message $s$, flip the last bit of $s$ to obtain a related
message $\tilde{s}$, and then re-encode $\tilde{s}$), and in most
interesting cases also allow the encoding to be randomized.  Formally,
a (binary) non-malleable code against a family of tampering functions
$\calF$ each mapping $\{0,1\}^n$ to $\{0,1\}^n$, consists of a
randomized encoding function $\enc : \{0,1\}^k \to \{0,1\}^n$ and a
deterministic decoding function $\dec : \{0,1\}^n \to \{0,1\}^k \cup
\{\perp\}$ (where $\perp$ denotes error-detection) which satisfy
$\dec(\enc(s))=s$ always, and the following non-malleability property
with error $\eps$: For every message $s \in \{0,1\}^k$ and every
function $f \in \calF$, the distribution of $\dec(f(\enc(s))$ is
$\eps$-close to a distribution $\cD_f$ that depends only on $f$ and is
independent of $s$ (ignoring the issue that $f$ may have too many fixed points). 


If some code enables error-detection against some family $\cF$, for
example if $\cF$ is the family of functions that flips between $1$ and
$t$ bits and the code has minimum distance more than $t$, then the
code is also non-malleable (by taking $\cD_f$ to be supported entirely
on $\perp$ for all $f$). Error-detection is also possible against the
family of ``additive errors," namely $\calF_{\mathsf{add}} = \{
f_\Delta \mid \Delta \in \{0,1\}^n \}$ where $f_\Delta(x) := x
+ \Delta$ (the addition being bit-wise XOR). Cramer et al. \cite{ref:CDFPW08} constructed ``Algebraic
Manipulation Detection" (AMD) codes of rate approaching $1$ such that
offset by an arbitrary $\Delta \neq 0$ will be detected with high
probability, thus giving a construction of non-malleable codes against
$\calF_{\mathsf{add}}$.

The notion of non-malleable coding becomes more interesting for
families against which error-detection is not possible.  A simple
example of such a class consists of all constant functions $f_{c}(x)
:= c$ for $c \in \{0,1\}^n$. Since the adversary can map all inputs to
a valid codeword $c^\ast$, one cannot in general detect tampering in
this situation. However, non-malleability is trivial to achieve in
this case as the output distribution of a constant function is
trivially independent of the message (so the rate $1$ code with
identity encoding function is itself non-malleable).

The original work~\cite{ref:nmc} showed that non-malleable codes of
positive rate exist against {\em every} not-too-large family $\cF$ of
tampering functions, specifically with $|\cF| \le 2^{2^{\alpha n}}$
for some constant $\alpha < 1$. In a companion paper~\cite{ref:CG1},
we proved that in fact one can achieve a rate approaching $1-\alpha$
against such families, and this is best possible in that there are
families of size $\approx 2^{2^{\alpha n}}$ for which non-malleable
coding is not possible with rate exceeding $1-\alpha$. (The latter is
true both for random families as well as natural families such as
functions that only tamper the first $\alpha n$ bits of the codeword.)

\subsection{Our results}
This work is focused on two natural families of tampering functions that have been studied in the literature. 
\vspace{-1ex}
\subsubsection{Bit-tampering functions}
The first class consists of {\em bit-tampering functions} $f$ in which
the different bits of the codewords are tampered independently
(i.e., each bit is either flipped, set to $0/1$, or left unchanged,
independent of other bits); formally $f(x) =
(f_1(x_1),f_2(x_2),\dots,f_n(x_n))$, where $f_1, \ldots, f_n\colon \zo
\to \zo$. As this family is ``small'' (of size $4^n$), by the above
general results, it admits non-malleable codes with positive rate, in
fact rate approaching $1$ by our recent result~\cite{ref:CG1}.

Dziembowski et al.~\cite{ref:nmc} gave a Monte Carlo construction of a
non-malleable code against this family; i.e., they gave an efficient
randomized algorithm to produce the code along with efficient encoding
and decoding functions such that w.h.p the encoder/decoder pair
ensures non-malleability against all bit-tampering functions. The rate
of their construction is, however, close to $.1887$ and thus falls short
of the ``capacity" (best possible rate) for this family of tampering
functions, which we now know equals $1$.

Our main result in this work is the following:
\begin{thm}
\label{thm:main-intro}
For all integers $n \ge 1$,
there is an explicit (deterministic) construction, with efficient encoding/decoding procedures, of a non-malleable code against bit-tampering functions that achieves rate $1-o(1)$ and error at most $\exp(-n^{\Omega(1)})$. 

If we seek error that is $\exp(-\tilde{\Omega}(n))$, we can guarantee that with an efficient Monte Carlo construction of the code that succeeds with probability $1-\exp(-\Omega(n))$.
\end{thm}

The basic idea in the above construction (described in detail in
Section~\ref{sec:explicit:construction}) is to use a concatenation
scheme with an outer code of rate close to $1$ that has large relative
distance and large dual relative distance, and as (constant-sized)
inner codes the non-malleable codes guaranteed by the existential
result (which may be deterministically found by brute-force if
desired).  This is inspired by the classical constructions of
concatenated codes \cite{ref:forney, ref:justesen}. The outer code
provides resilience against tampering functions that globally fix too
many bits or alter too few.  For other tampering functions, in order
to prevent the tampering function from locally freezing many entire
inner blocks (to possibly wrong inner codewords), the symbols of the
concatenated codeword are permuted by a {\em pseudorandom
permutation\footnote{Throughout the paper,
by pseudorandom permutation we mean $t$-wise independent
permutation (as in Definition~\ref{def:limited:perm}) for an appropriate choice of $t$. This should not be confused
with cryptographic pseudorandom permutations, which are not used in this work.
}.}

The seed for the permutation is itself included as the
initial portion of the final codeword, after encoding by a
non-malleable code (of possibly low rate). This protects the seed and
ensures that any tampering of the seed portion results in the decoded
permutation being essentially independent of the actual permutation,
which then results in many inner blocks being error-detected (decoded
to $\perp$) with noticeable probability each. The final decoder
outputs $\perp$ if any inner block is decoded to $\perp$, an event
which happens with essentially exponentially small probability in $n$
with a careful choice of the parameters. The above scheme uses
non-malleable codes in two places to construct the final non-malleable
code, but there is no circularity because the codes for the inner
blocks are of constant size, and the code protecting the seed can have
very low rate (even sub-constant) as the seed can be made much smaller
than the message length.

The structure of our construction bears some high level similarity to
the optimal rate code construction for correcting a bounded number of
additive errors in \cite{ref:GS10}. The exact details though are quite
different; in particular, the crux in the analysis of \cite{ref:GS10}
was ensuring that the decoder can recover the seed correctly, and
towards this end the seed's encoding was distributed at random
locations of the final codeword. Recovering the seed is both
impossible and not needed in our context here.
\subsubsection{Split-state adversaries}
Bit-tampering functions act on different bits independently. A much
more general class of tampering functions considered in the
literature~\cite{ref:nmc,ref:DKO,ref:ADL} is the so-called {\em
split-state model}. Here the function $f : \{0,1\}^n \to \{0,1\}^n$
must act on each half of the codeword independently (assuming $n$
is even), but can act
arbitrarily within each half. Formally, $f(x) = (f_1(x_1),f_2(x_2))$
for some functions $f_1,f_2 : \{0,1\}^{n/2} \to \{0,1\}^{n/2}$ where
$x_1,x_2$ consist of the first $n/2$ and last $n/2$ bits of $x$. This
represents a fairly general and useful class of adversaries which are
relevant for example when the codeword is stored on two physically
separate devices, and while each device may be tampered arbitrarily,
the attacker of each device does not have access to contents stored on
the other device.

The capacity of non-malleable coding in the split-state model equals
$1/2$, as established in our recent work~\cite{ref:CG1}. A natural
question therefore is to construct {\em efficient} non-malleable codes
of rate approaching $1/2$ in the split-state model (the results
in \cite{ref:nmc} and \cite{ref:CG1} are existential, and the codes do
not admit polynomial size representation or polynomial time
encoding/decoding). This remains a challenging open question, and in
fact constructing a code of positive rate itself seems rather
difficult. A code that encodes one-bit messages is already
non-trivial, and such a code was constructed in \cite{ref:DKO} by
making a connection to two-source extractors with sufficiently strong
parameters and then instantiating the extractor with a construction
based on the inner product function over a finite field. We stress
that this connection to two-source extractor only applies to encoding
one-bit messages, and does not appear to generalize to longer
messages.

Recently, Aggarwal, Dodis, and Lovett~\cite{ref:ADL} solved the
central open problem left in \cite{ref:DKO} --- they construct a
non-malleable code in the split-state model that works for arbitrary
message length, by bringing to bear elegant techniques from additive
combinatorics on the problem. The rate of their code is polynomially
small: $k$-bit messages are encoded into codewords with $n \approx
k^7$ bits.

In the second part of this paper (Section~\ref{sec:nmext}), we study the problem of
non-malleable coding in the split-state model. We do not offer any
explicit constructions, and the polynomially small rate achieved
in \cite{ref:ADL} remains the best known. Our contribution here is
more conceptual. We define the notion of non-malleable two-source
extractors, generalizing the influential concept of non-malleable
extractors introduced by Dodis and Wichs~\cite{ref:DW09}. A
non-malleable extractor is a regular seeded extractor $\mathsf{Ext}$
whose output $\mathsf{Ext}(X,S)$ on a weak-random source $X$ and
uniform random seed $S$ remains uniform even if one knows the value
$\mathsf{Ext}(X,f(S))$ for a related seed $f(S)$ where $f$ is a
tampering function with no fixed points. In a two-source non-malleable
extractor we allow both sources to be weak and independently tampered, and we
further extend the definition to allow the functions to have fixed
points in view of our application to non-malleable codes.
We prove, however, that for construction of two-source non-malleable extractors,
it suffices to only consider tampering functions that have no fixed points,
at cost of a minor loss in the parameters.

We show that given a two-source non-malleable extractor $\nm$ with
exponentially small error in the output length, one can build a
non-malleable code in the split-state model by setting the extractor
function $\nm$ to be the decoding function (the encoding of $s$
then picks a pre-image in $\nm^{-1}(s)$).  

This identifies a possibly natural avenue to construct improved
non-malleable codes against split-state adversaries by constructing
non-malleable two-source extractors, which seems like an interesting
goal in itself.  Towards confirming that this approach has the
potential to lead to good non-malleable codes, we prove a fairly
general existence theorem for seedless non-malleable extractors, by essentially
observing that the ideas from the proof of existence of seeded
non-malleable extractors in \cite{ref:DW09} can be applied in a much
more general setting. 
Instantiating this result with split-state
tampering functions, we show the existence of non-malleable two-source
extractors with parameters that are strong enough to imply
non-malleable codes of rate arbitrarily close to $1/5$ in the split-state
model. 

Explicit construction of (ordinary) two-source extractors and closely-related
objects is a well-studied problem in the literature and an abundance of explicit constructions for
this problem is known\footnote{Several of these constructions are structured enough
to easily allow for efficient sampling of a uniform pre-image from $\mathsf{Ext}^{-1}(s)$.} 
(see, e.g., \cite{ref:BRSW06, ref:Bou05, ref:CG, 
ref:KLR09, ref:Rao08, ref:Raz05}). 
The problem becomes increasingly challenging, however, (and remains open to date)
when the entropy rate of the two sources may be noticeably below $1/2$.
Fortunately, we show that for construction of constant-rate non-malleable codes
in the split-state model, it suffices to have two-source
non-malleable extractors for source entropy rate $.99$ and with
some output length $\Omega(n)$ (against tampering functions with no 
fixed points). Thus the infamous ``$1/2$ entropy rate barrier'' on
two-source extractors does not concern our particular application.

Furthermore, we note that for seeded non-malleable
extractors (which is a relatively recent notion) there are already a few exciting explicit
constructions~\cite{ref:DLWZ11,ref:GRS12,ref:Li12}\footnote{\cite{ref:Li12}
also establishes a connection between seeded non-malleable
extractors and ordinary two-source extractors.}.
The closest construction to our application is
\cite{ref:DLWZ11} which is in fact
a two-source non-malleable extractor when the adversary may tamper
with either of the two sources (but not simultaneously both).
Moreover, the coding scheme defined by this extractor (which is
the character-sum extractor of Chor and Goldreich~\cite{ref:CG}) naturally
allows for an efficient encoder and decoder.
Nevertheless, it appears challenging
to extend known constructions of seeded non-malleable extractors
to the case when both inputs can be tampered.
We leave explicit constructions of non-malleable two-source
extractors, even with sub-optimal parameters, as an interesting open
problem for future work. 


\section{Preliminaries}

\subsection{Notation}
We use $\U_n$ for the uniform distribution on $\zo^n$ and $U_n$
for the random variable sampled from $\U_n$ and independently 
of any existing randomness.
For a random variable $X$, we denote by $\distr(X)$ the probability distribution
that $X$ is sampled from.
Generally, we will use calligraphic symbols (such as $\cX$) for probability distributions
and the corresponding capital letters (such as $X$) for related random variables.
We use $X \sim \cX$ to denote that the random variable $X$ is drawn from 
the distribution $\cX$.
Two distributions $\cX$ and $\cY$ being $\eps$-close in statistical distance is denoted by
$\cX \approx_\eps \cY$. We will use $(\cX, \cY)$ for the product distribution
with the two coordinates independently sampled from $\cX$ and $\cY$.
All unsubscripted logarithms are taken to the base $2$.
Support of a discrete random variable $X$ is denoted by $\supp(X)$.
A distribution is said to be \emph{flat} if it is uniform on its support.
For a sequence $x = (x_1, \ldots, x_n)$ and set $S \subseteq [n]$, we use
$x|_S$ to denote the restriction of $x$ to the coordinate positions chosen by $S$.
We use $\tilde{O}(\cdot)$ and $\tilde{\Omega}(\cdot)$ to denote asymptotic
estimates that hide poly-logarithmic factors in the involved parameter.

\subsection{Definitions}
In this section, we review the formal definition of non-malleable codes as introduced
in \cite{ref:nmc}. First, we recall the notion of \emph{coding schemes}.

\begin{defn}[Coding schemes] \label{def:scheme}
A pair of functions $\enc\colon \zo^k \to \zo^n$ and $\dec\colon \zo^n \to \zo^k \cup \{\perp\}$
where $k \leq n$ 
is said to be a coding scheme with block length $n$ and message length $k$
if the following conditions hold.
\begin{enumerate}
\itemsep=0ex
\vspace{-1ex}
\item The encoder $\enc$ is a randomized function; i.e., at each call it receives a 
uniformly random sequence of coin flips that the output may depend on. This random input
is usually omitted from the notation and taken to be implicit. Thus for any 
$s \in \zo^k$, $\enc(s)$ is a random variable over $\zo^n$. The decoder $\dec$ is;
however, deterministic.

\item For every $s \in \zo^k$, we have $\dec(\enc(s)) = s$ with probability $1$.
\end{enumerate}

The \emph{rate} of the coding scheme is the ratio $k/n$.
A coding scheme is said to have relative distance $\delta$ (or minimum distance
$\delta n$), for some $\delta \in [0,1)$,
if for every $s \in \zo^k$ the following holds. Let $X := \enc(s)$. Then,
for any $\Delta \in \zo^n$ of Hamming weight at most $\delta n$,
$\dec(X + \Delta) = \perp$ with probability $1$. \qed
\end{defn}

\noindent Before defining non-malleable coding schemes, we find it convenient to define the following notation.

\begin{defn}
For a finite set $\Gamma$, the function $\Copy\colon (\Gamma \cup \{\same\}) \times \Gamma \to 
\Gamma$ is defined as follows:
\[
\Copy(x, y) := \begin{cases}
x & x \neq \same, \\
y & x = \same.
\end{cases} \qquad\qquad\qed
\]
\end{defn}

\noindent
The notion of non-malleable coding schemes from \cite{ref:nmc}
can now be rephrased as follows.

\begin{defn}[Non-malleability] \label{def:nmCode}
A coding scheme $(\enc, \dec)$ with message length $k$ and block length $n$
is said to be non-malleable with error $\eps$ (also called \emph{exact security})
with respect to a family
$\cF$ of tampering functions acting on $\zo^n$ (i.e., each $f \in \cF$ maps
$\zo^n$ to $\zo^n$) if for every $f \in \cF$ there is a distribution
$\cD_f$ over $\zo^k \cup \{\perp, \same\}$ such that the following holds for
all $s \in \zo^k$.
Define the random variable \[S := \dec(f(\enc(s))),\] and 
let $S'$ be independently sampled from $\cD_f$. Then,
\[
\distr(S) \approx_\eps \distr(\Copy(S', s)). 
\] \qed
\end{defn}

%

\begin{remark}[Efficiency of sampling $\cD_f$] \label{rem:efficiency}
The original definition of non-malleable codes in \cite{ref:nmc} also requires
the distribution $\cD_f$ to be efficiently samplable given oracle access
to the tampering function $f$. It should be noted; however, that
for any non-malleable coding scheme equipped with an efficient encoder
and decoder, it can be shown that the following is a valid and efficiently samplable 
choice for the distribution $\cD_f$
(possibly incurring a constant factor increase in the error parameter):  
\begin{enumerate}
\itemsep=0ex
\vspace{-1ex}
\item Let $S \sim \U_k$, and $X := f(\enc(S))$.
\item If $\dec(X) = S$, output $\same$. Otherwise, output $\dec(X)$.
\end{enumerate}
\end{remark}

\begin{defn}[Sub-cube]
A sub-cube over $\zo^n$ is a set
$S \subseteq \zo^n$ such that for some $T = \{ t_1, \ldots, t_\ell \} \subseteq [n]$
and $w = (w_1, \ldots, w_\ell) \in \zo^\ell$,
\[
S = \{ (x_1, \ldots, x_n) \in \zo^n\colon x_{t_1} = w_1, \ldots, x_{t_\ell} = w_\ell\};
\]
the $\ell$ coordinates in $T$ are said to be {\em frozen} and the remaining $n-\ell$ are said to be random.
\end{defn}

Throughout the paper, we use the following notions of limited independence.

\begin{defn}[Limited independence of bit strings] \label{def:limited:string}
A distribution $\cD$ over $\zo^n$ is said to be \emph{$\ell$-wise $\delta$-dependent}
for an integer $\ell > 0$ and parameter $\delta \in [0, 1)$ if the marginal distribution
of $\cD$ restricted to any subset $T \subseteq [n]$ of the coordinate positions
where $|T| \leq \ell$ is $\delta$-close to $\U_{|T|}$. When $\delta = 0$, the distribution
is $\ell$-wise independent.
\end{defn}

\begin{defn}[Limited independence of permutations] \label{def:limited:perm}
The distribution of a random permutation $\Pi\colon [n] \to [n]$ is 
said to be \emph{$\ell$-wise $\delta$-dependent}
for an integer $\ell > 0$ and parameter $\delta \in [0, 1)$ if for every $T \subseteq [n]$
such that $|T| \leq \ell$, the marginal distribution of the sequence $(\Pi(t)\colon t \in T)$
is $\delta$-close to that of $(\bar{\Pi}(t)\colon t \in T)$, where $\bar{\Pi}\colon [n] \to [n]$
is a uniformly random permutation.
\end{defn}

We will use the following notion of \emph{Linear Error-Correcting Secret Sharing Schemes}
(LECSS) as formalized by Dziembowski et al.~\cite{ref:nmc} for their construction of non-malleable coding
schemes against bit-tampering adversaries.

\begin{defn}[LECSS] \cite{ref:nmc} \label{def:lecss}
A coding scheme $(\enc, \dec)$ of block length $n$ and message length $k$ is a
$(d, t)$-\emph{Linear Error-Correcting Secret Sharing Scheme} (LECSS),
for integer parameters $d, t \in [n]$ if 
\begin{enumerate}
\item The minimum distance of the coding
scheme is at least $d$, 
\item For every message $s \in \zo^k$, the distribution of $\enc(s) \in \zo^n$ is $t$-wise independent
(as in Definition~\ref{def:limited:string}).

\item For every $w, w' \in \zo^n$ such that $\dec(w) \neq \perp$ and\footnote{
Although we use LECSS codes in our explicit construction, contrary to 
\cite{ref:nmc} we do not directly use the linearity of the code for our proof.}
$\dec(w') \neq \perp$, we have $\dec(w+w') = \dec(w) + \dec(w')$,
where we use bit-wise addition over $\F_2$.
\end{enumerate}
\end{defn}

\section{Existence of optimal bit-tampering coding schemes}
\label{sec:random}

In this section, we recall the probabilistic construction of non-malleable codes
introduced in \cite{ref:CG1}. This construction, depicted as Construction~\ref{constr:prob},
is defined with respect to an integer parameter
$t > 0$ and a \emph{distance parameter} $\delta \in [0, 1)$.

\begin{constr} 
  \caption{Probabilistic construction of non-malleable codes in \cite{ref:CG1}.}

  \begin{itemize}
  \item {\it Given: } Integer parameters $0 < k \leq n$ and integer $t > 0$
  such that $t 2^k \leq 2^n$, and a distance parameter $\delta \geq 0$.

  \item {\it Output: } A pair of functions $\enc\colon \zo^k \to \zo^n$
  and $\dec\colon \zo^n \to \zo^k$, where $\enc$ may also use a uniformly random
  seed which is hidden from that notation, but $\dec$ is deterministic.

  \item {\it Construction: } 
   \begin{enumerate}
  \item Let $\cN := \zo^n$. 
  \item For each $s \in \zo^k$, in an arbitrary order, 
  \begin{itemize}
  \item Let $E(s) := \emptyset$.
  \item For $i \in \{1, \ldots, t\}$:
  \begin{enumerate}
  \item Pick a uniformly random vector $w \in \cN$.
  \item Add $w$ to $E(s)$.
  \item Let $\Gamma(w)$ be the Hamming ball of radius $\delta n$ centered at $w$.
  Remove $\Gamma(w)$ from $\cN$ (note that when $\delta = 0$, we have $\Gamma(w) = \{w\}$).
  \end{enumerate}
  \end{itemize}
  \item Given $s \in \zo^k$, $\enc(s)$ outputs an element of $E(s)$ uniformly
  at random.
  
  \item Given $w \in \zo^n$, $\dec(s)$ outputs the unique $s$ such that
  $w \in E(s)$, or $\perp$ if no such $s$ exists.
\end{enumerate}     
  \end{itemize}
  \label{constr:prob}
\end{constr}

The following, proved in \cite{ref:CG1}, shows non-malleability of the construction.

\begin{thm}[\cite{ref:CG1}] \label{thm:upperBound}
Let $\cF\colon \zo^n \to \zo^n$ be any family of tampering functions.
For any $\eps, \eta > 0$, with probability at least $1-\eta$,
the coding scheme $(\enc, \dec)$ of Construction~\ref{constr:prob}
is a 
non-malleable code with respect to $\cF$ and with error $\eps$
and relative distance $\delta$,
provided that both of the following conditions are satisfied.
\begin{enumerate}
\itemsep=0ex
\item $t \geq t_0$, for some 
\begin{equation} \label{eqn:thm:upper:t}
t_0 = O\left( \frac{1}{\eps^6} \Big(\log\frac{|\cF| 2^n}{\eta} \Big) \right).
\end{equation}
\item $k \leq k_0$, for some 
\begin{equation} \label{eqn:thm:upper:k0}
k_0 \geq n(1-h(\delta))-\log t-3\log(1/\eps)-O(1),
\end{equation}
where $h(\cdot)$ denotes the binary entropy function.
\end{enumerate}
\end{thm}

\begin{remark} \label{rem:Df}
The proof of Theorem~\ref{thm:upperBound} explicitly defines the 
choice of $\cD_f$ of Definition~\ref{def:nmCode} to be the distribution
of the following random variable:
\begin{equation} \label{eqn:explicit:Df}
D := \begin{cases}
\same & \text{if $f(U_n) = U_n$}, \\
\dec(f(U_n)) & \text{if $f(U_n) \neq U_n$ and $f(U_n) \in H$}, \\
\perp & \text{otherwise,}
\end{cases}
\end{equation}
where $H \subseteq \zo^n$ is the set 
\begin{equation}
\label{eqn:explicit:H}
H := \{ x \in \zo^n\colon \Pr[f(U_n) = x] > 1/r \},
\end{equation}
for an appropriately chosen $r = \Theta(\eps^2 t)$.
\end{remark}

We now instantiate the above result to the specific case of bit-tampering adversaries,
and derive additional properties of the coding scheme of Construction~\ref{constr:prob}
that we will later use in our explicit construction.

\begin{lem}(Cube Property) \label{lem:cube}
Consider the coding scheme $(\enc, \dec)$ of Construction~\ref{constr:prob}
with parameters $t$ and $\delta$,
and assume that $t 2^{k-n(1-h(\delta))} \leq 1/8$, where $h(\cdot)$
is the binary entropy function.
Then, there is a $\delta_0 = O(\log n / n)$ such that if
$\delta \geq \delta_0$, the following holds
with probability at least $1-\exp(-n)$ over the randomness of the code
construction.
For any sub-cube $S \subseteq \zo^n$ of size at least $2$,
and $U_S \in \zo^n$ taken uniformly at random from $S$,
\[
\Pr_{U_S}[\dec(U_S) = \perp] \geq 1/2.
\]
\end{lem}

\begin{proof}
Let $S \subseteq \zo^n$ be any sub-cube, and let $\gamma := tK/2^n$, where $K := 2^k$. The assumption implies that $\gamma V \leq 1/8$, where
$V \leq 2^{nh(\delta)}$ is the volume of a Hamming ball of radius $\delta n$.
Let $E_1, \ldots, E_{tK}$ be the codewords chosen by the code construction
in the order they are picked.

If $|S| \geq 2 tK$, the claim obviously holds (since the total number of 
codewords in $\supp(\enc(\U_k))$ is $tK$, thus we can assume otherwise.

Arbitrarily order the elements of $S$ as $s_1, \ldots, s_{|S|}$, and for each
$i \in [|S|]$, let the indicator random variable $X_i$ be so that 
$X_i = 1$ iff $\dec(s_i) \neq \perp$. Define $X_0 = 0$. Our goal is to upper bound
\[
\E[X_i | X_0, \ldots, X_{i-1}]
\]
for each $i \in [|S|]$. Instead of conditioning on $X_1, \ldots, X_{i-1}$,
we condition on a more restricted event and show that regardless of the
more restricted conditioning, the expectation of $X_i$ can still be upper bounded as desired.
Namely, we condition on the knowledge of not only $\dec(s_j)$ for all $j<i$ but also the
unique $j' \in [tK]$ such that $E_{j'} = s_j$, if $\dec(s_j) \neq \perp$.
Obviously the knowledge of this information determines the values of 
$X_1, \ldots, X_{i-1}$, and thus Proposition~\ref{prop:restriction} applies.
Under the more restricted conditioning, some of the codewords in $E_{1}, \ldots, E_{tK}$ (maybe all) will
be revealed. Obviously, the revealed codewords have no chance of being assigned
to $s_i$ (since the codewords are picked without replacement). By a union bound, the chance that any of the up to $tK$
remaining codewords is assigned to $s_i$ by the decoder is thus at most
\[
\frac{tK}{2^n - |S| V} \leq
\frac{tK}{2^n(1-2 \gamma V)} \leq (4/3) tK/2^n = (4/3) \gamma \leq 1/6.
\]
Since the above holds for any realization of the information that
we condition on, we conclude that
\[
\E[X_i | X_0, \ldots, X_{i-1}] \leq 1/6.
\]
Let $X := X_1 + \cdots + X_{|S|}$, which determines the number of 
vectors in $S$ that are hit by the code.
We can apply Proposition~\ref{prop:simpleAzuma} to deduce that
\[
\Pr[X > |S|/2] \leq \exp(-|S|/18).
\]
Therefore, if $|S| > S_0$ for some $S_0 = O(n)$, the upper bound can be
made less than $\exp(-n) 3^{-n}$. In this case, a union bound on all possible
sub-cubes satisfying the size lower bound ensures that the desired cube property
holds for all such sub-cubes with probability at least $1-\exp(-n)$.

The proof is now reduced to sub-cubes with at most $\delta_0 n = O(\log n)$ random bits,
where we choose $\delta_0 := (\log S_0)/n$. 
In this case, since the relative distance of the coding scheme of Construction~\ref{constr:prob}
is always at least $\delta \geq \delta_0$, we deduce that 
\[
|\{x \in S\colon \dec(x) \neq \perp\}| \leq 1 \leq |S|/2,
\]
where the first inequality is due to the minimum distance of the code
and the second is due to the assumption that $|S| \geq 2$. Thus,
whenever $2 \leq |S| \leq S_0$, we always have the property that
\[
\Pr_{U_S}[\dec(U_S = \perp)] \geq 1/2. \qedhere
\]
\end{proof}

\begin{lem}(Bounded Independence) \label{lem:boundedIndep}
Let $\ell \in [n]$, $\eps > 0$ and suppose the parameters are as in Construction~\ref{constr:prob}.
Let $\gamma := t2^{k-n(1-h(\delta))}$, where $h(\cdot)$ denotes the binary entropy
function. There is a choice of 
\[
t_0 = O\Big( \frac{2^\ell + n}{{\eps}^2} \Big)
\]
such that, provided that $t \geq t_0$,
with probability $1-\exp(-n)$ over the randomness of the code construction
the coding scheme $(\enc, \dec)$ satisfies the following: For any $s \in \zo^k$,
the random vector $\enc(s)$ is $\ell$-wise $\eps'$-dependent,
where
\[
\eps' := \max\Big\{ \eps, \frac{2 \gamma}{1-\gamma} \Big\}.
\]
\end{lem}
\begin{proof}
Consider any message $s \in \zo^k$ and suppose the $t$ codewords in 
$\supp(\enc(s))$ are denoted by $E_1, \ldots, E_t$ in the
order they are picked by the construction. 

Let $T \subseteq [n]$ be any set of size at most $\ell$.
Let $E'_1, \ldots, E'_t \in \zo^{|T|}$ be the restriction of $E_1, \ldots, E_t$
to the coordinate positions picked by $T$.
Observe that the distribution of $\enc(s)$ restricted to the
coordinate positions in $T$ is exactly the empirical
distribution of the vectors $E'_1, \ldots, E'_t$, and the support size
of this distribution is bounded by $2^\ell$.

Let $K := 2^k$, $N := 2^n$, and $V \leq 2^{n h(\delta)}$ be the 
volume of a Hamming ball of radius $\delta n$.
By the code construction, for $i \in [t]$, conditioned on the knowledge of
$E_1, \ldots, E_{i-1}$, the distribution of $E_i$ is uniform on
$\zo^n \setminus (\Gamma(E_1) \cup \ldots \cup \Gamma(E_{i-1}))$ which
is a set of size at least $N(1-tK V) \geq N(1-\gamma)$. By Proposition~\ref{prop:uniformity},
it follows that the conditional distribution of each $E_i$ remains
$(\gamma /(1-\gamma))$-close to $\U_n$.
Since the $E'_i$ are simply restrictions of
the $E_i$ to some subset of the coordinates, 
the same holds for the $E'_i$; i.e., the distribution of $E'_i$
conditioned on the knowledge of $E'_1, \ldots, E'_{i-1}$ is
$(\gamma /(1-\gamma))$-close to $\U_{|T|}$.

Observe that $\eps' - \gamma/(1-\gamma) \geq \eps'/2$.
By applying Lemma~\ref{lem:distrLearning:dependent} to the sample
outcomes $E'_1, \ldots, E'_{t}$, we can see that
with probability at least $\exp(-3n)$
over the code construction, the empirical distribution of the $E'_i$
is $\eps'$-close to uniform provided that $t \geq t_0$ for some
\[
t_0 = O\Big( \frac{2^\ell + n}{{\eps'}^2} \Big) = O\Big( \frac{2^\ell + n}{{\eps}^2} \Big).
\]
Now, we can take a union bound on all choices of the message $s$
and the set $T$ and obtain the desired conclusion.
\end{proof}

We now put together the above results to conclude our main existential result about the codes that we will use at the ``inner" level to encode blocks in our construction of non-malleable codes against bit tampering functions. Among the properties guaranteed below, we in fact do not need the precise non-malleability property (item~\ref{prop:nm} in the statement of Lemma~\ref{lem:properties} below) in our eventual proof, although we use non-malleability to prove the last property (item~\ref{prop:bottom}) which is needed in the proof.

\begin{lem} \label{lem:properties}
Let $\alpha > 0$ be any parameter. Then, there is an $n_0 = O(\log^2(1/\alpha)/\alpha)$
such that for any $n \geq n_0$, 
Construction~\ref{constr:prob} can be set up so that with probability
$1-3\exp(-n)$ over the randomness of the construction, the resulting coding scheme
$(\enc, \dec)$ satisfies the following properties:
\begin{enumerate}
\item (Rate) Rate of the code is at least $1-\alpha$. \label{prop:rate}

\item (Non-malleability) The code is non-malleable against bit-tampering adversaries with
error $\exp(-\Omega(\alpha n))$. \label{prop:nm}

\item (Cube property) The code satisfies the \emph{cube property} of Lemma~\ref{lem:cube}. \label{prop:cube}

\item (Bounded independence) For any message $s \in \zo^k$, the distribution of $\enc(s)$
is $\exp(-\Omega(\alpha n))$-close to an $\Omega(\alpha n)$-wise
independent distribution with uniform entries. \label{prop:indep}

\item (Error detection) Let $f\colon \zo^n \to \zo^n$ be any bit-tampering adversary that is
neither the identity function nor a constant function.
Then, for every message $s \in \zo^k$,
\[
\Pr[\dec(f(\enc(s))) = \perp] \geq 1/3,
\]
where the probability is taken over the randomness of the encoder. \label{prop:bottom}
\end{enumerate}
\end{lem}

\begin{proof}
Consider the family $\cF$ of bit-tampering functions, and observe that
$|\cF| = 4^n$. First, we apply Theorem~\ref{thm:upperBound} with
error parameter $\eps := 2^{-\alpha n/27}$, distance parameter
$\delta := h^{-1}(\alpha/3)$, and success parameter
$\eta := \exp(-n)$. Let $N := 2^n$ and observe that $\log (N|\cF|/\eta) = O(n)$.
We choose $t = \Theta (n/\eps^6)$ so as to ensure that the
coding scheme $(\enc, \dec)$ is non-malleable for bit-tampering adversaries
with error at most $\eps$, relative distance at least $\delta$, and message length
\[
k \geq n (1-h(\delta)) - 9 \log(1/\eps) - \log n - O(1)
\geq (1- 2\alpha/3) n - \log n - O(1),
\]
which can be made at least $n(1-\alpha)$ if $n \geq n_1$ for some
$n_1 = O(\log(1/\alpha)/\alpha)$. This ensures
that properties \ref{prop:rate}~and~\ref{prop:nm} are satisfied.

In order to ensure the cube property (property~\ref{prop:cube}), we can apply Lemma~\ref{lem:cube}.
Let $K := 2^k$ and note that our choices of the parameters imply $tK/N^{1-h(\delta)} = O(\eps^3) \ll 1/8$.
Furthermore, consider the parameter $\delta_0 = O((\log n)/n)$ of Lemma~\ref{lem:cube} and observe that
$\alpha/3 = h(\delta) = O(\delta \log(1/\delta))$. We thus see that as long as
$n \geq n_2$ for some $n_2 = O(\log^2(1/\alpha)/\alpha)$, we may ensure
that $\delta n \geq \delta_0 n$. By choosing $n_0 := \max\{ n_1, n_2 \}$, 
we see that the requirements of Lemma~\ref{lem:cube} is satisfied, implying that
with probability at least $1-\exp(-n)$, the cube property is satisfied.

As for the bounded independence property (Property~\ref{prop:indep}), consider
the parameter $\gamma$ of Lemma~\ref{lem:boundedIndep} and recall that we have shown
$\gamma = O(\eps^3)$. Thus by Lemma~\ref{lem:boundedIndep}, with probability
at least $1-\exp(-n)$, every encoding $\enc(s)$ is $\ell$-wise $\sqrt{\eps}$-dependent
for some
\begin{equation}
\label{eqn:property:bounded:ell}
\ell \geq \log t - 2 \log(1/\sqrt{\eps}) - \log n - O(1) \geq  5 \log(1/\eps) - O(1)
= \Omega(\alpha n).
\end{equation}

Finally, we show that property~\ref{prop:bottom} is implied by
properties \ref{prop:nm}, \ref{prop:cube}, and \ref{prop:indep}
that we have so far shown to simultaneously hold with probability 
at least $1-3\exp(-n)$.
In order to do so, we first recall that Theorem~\ref{thm:upperBound} explicitly defines the choice of $\cD_f$ in
Definition~\ref{def:nmCode} according to \eqref{eqn:explicit:Df}.
Let $H \subseteq \zo^n$ be the set of heavy elements as in \eqref{eqn:explicit:H}
and $r = \Theta(\eps^2 t)$ be the corresponding threshold parameter in
the same equation. Let $f\colon \zo^n \to \zo^n$ be any non-identity bit-tampering function
and let $\ell' \in [n]$ be the number of bits that are either flipped or left
unchanged by $f$. We consider two cases.
\begin{description}
\item[Case~1: $\ell' \geq \log r$.] In this case, for every $x \in \zo^n$, we have
\[
\Pr[f(\U_n) = x] \leq 2^{-\ell'} \leq r,
\]
and thus $H = \emptyset$. Also observe that, for $U \sim \U_n$,
\[
\Pr[f(U) = U] \leq 1/2,
\]
the maximum being achieved when $f$ freezes only one bit and leaves
the remaining bits unchanged (in fact,
if $f$ flips any of the bits, the above probability becomes zero).

We conclude that in this case, the entire probability mass of $\cD_f$
is supported on $\{ \same, \perp \}$ and the mass assigned to $\same$
is at most $1/2$. Thus, by definition of non-malleability, for every message $s \in \zo^k$,
\[
\Pr[\dec(f(\enc(s))) = \perp] \geq 1/2 - \eps \geq 1/3.
\]

\item[Case~2: $\ell' < \log r$.] 
Since $r = \Theta(\eps^2 t)$, by plugging in the value of $t$ we see that
$r = O(n/\eps^4)$, and thus we know that $\ell' < \log n + 4 \log(1/\eps) + O(1)$.

Consider any $s \in \zo^k$, and recall
that, by the bounded independence property, we already know that $\enc(s)$
is $\ell$-wise $\sqrt{\eps}$-dependent. Furthermore, 
by \eqref{eqn:property:bounded:ell}, 
\[
\ell \geq 5 \log(1/\eps) - O(1) \geq \ell',
\]
where the second inequality follows by the assumed lower bound $n \geq n_0$
on $n$. We thus can use the $\ell$-wise independence property of $\enc(s)$
and deduce that the distribution of $f(\enc(s))$ is $(\sqrt{\eps})$-close
to the uniform distribution on a sub-cube $S \subseteq \zo^n$ of size at least
$2$. Combined with the cube property (property~\ref{prop:cube}), we see
that
\[
\Pr[\dec(f(\enc(s))) = \perp] \geq 1/2 - \sqrt{\eps} \geq 1/3.
\]

Finally, by applying a union bound on all the failure probabilities, we conclude
that with probability at least $1-3\exp(-n)$, the code resulting from Construction~\ref{constr:prob}
satisfies all the desired properties.
\end{description}
\end{proof}

\section{Explicit construction of optimal bit-tampering coding schemes}
\label{sec:explicit}

In this section, we describe an explicit construction of codes achieving 
rate close to $1$ that are non-malleable against bit-tampering adversaries.
Throughout this section, we use $N$ to denote the block length of the
final code.

\subsection{The construction}
\label{sec:explicit:construction}

At a high level, we combine the following tools in our construction: 
1) an inner code $\cC_0$ (with encoder $\enc_0$) 
of constant length satisfying the properties of Lemma~\ref{lem:properties};
2) an existing non-malleable code construction $\cC_1$ (with encoder $\enc_1$)
against bit-tampering achieving a 
possibly low (even sub-constant) rate; 3) a linear error-correcting secret sharing
scheme (LECSS) $\cC_2$ (with encoder $\enc_2$); 4) an explicit function $\Perm$ that, given a uniformly
random seed, outputs a pseudorandom permutation
(as in Definition~\ref{def:limited:perm}) on a
domain of size close to $N$. Figure~\ref{fig:explicit} depicts how various
components are put together to form the final code construction.

At the outer layer, LECSS is used to pre-code the message. The resulting string
is then divided into blocks, where each block is subsequently encoded by
the inner encoder $\enc_0$. For a ``typical'' adversary that flips or 
freezes a prescribed fraction of the bits, we expect many of the inner
blocks to be sufficiently tampered so that many of the inner blocks
detect an error when the corresponding inner decoder is called. However,
this ideal situation cannot necessarily be achieved if the
fraction of global errors is too small, or if too many bits are frozen 
by the adversary (in particular, the adversary may freeze all but few of the
blocks to valid inner codewords). In this case, we rely on distance and bounded independence
properties of LECSS to ensure that the outer decoder, given the
tampered information, either detects
an error or produces a distribution that is independent of the source
message.

A problem with the above approach is that the adversary knows the location
of various blocks, and may carefully design a tampering scheme that,
for example, freezes a large fraction of the blocks to valid inner codewords and
leaves the rest of the blocks intact. To handle adversarial strategies
of this type, we permute the final codeword using the pseudorandom permutation
generated by $\Perm$, and include the seed in the final codeword. Doing this
has the effect of randomizing the action of the adversary, but on the other hand creates the
problem of protecting the seed against tampering. In order to solve this problem,
we use the sub-optimal code $\cC_1$
to encode the seed and prove in the analysis that non-malleability of the code $\cC_1$
can be used to make the above intuitions work.

\begin{figure}
\centering
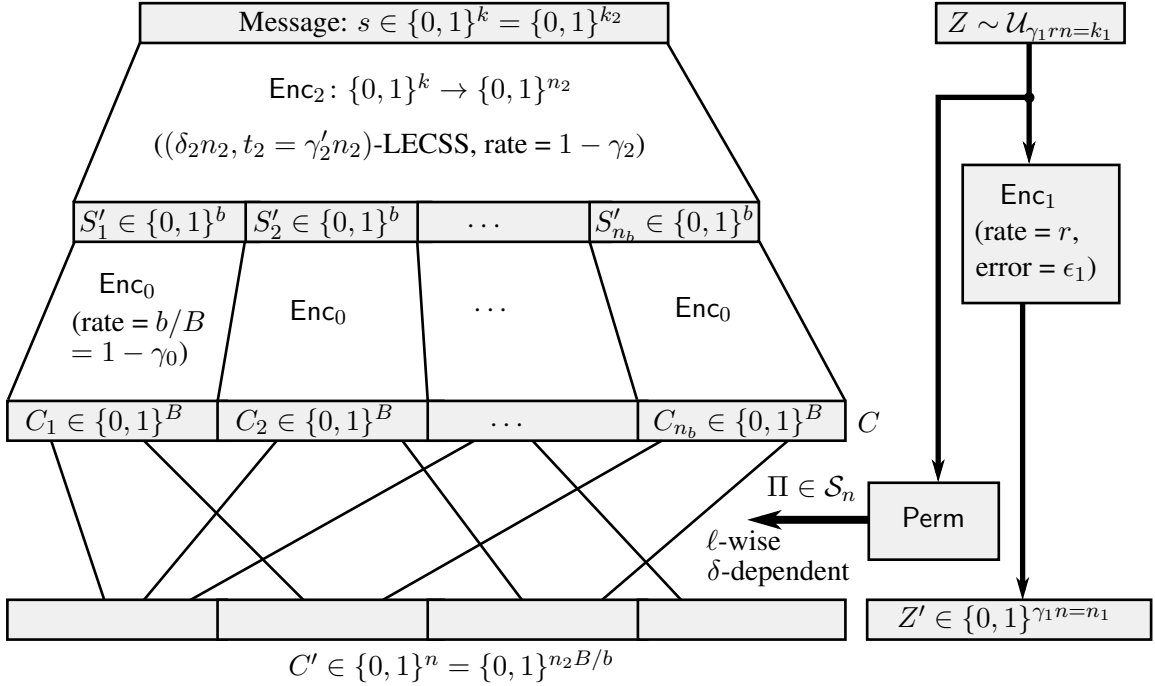
\caption{Schematic description of the encoder $\enc$ from our explicit construction.}
\label{fig:explicit}
\end{figure}

\subsubsection{The building blocks}
\label{sec:blocks}

In the construction, we use the following building blocks, with some of
the parameters to be determined later in the analysis.

\begin{enumerate}
\item 
An inner coding scheme $\cC_0=(\enc_0, \dec_0)$ with
rate $1-\gamma_0$ (for an arbitrarily small parameter $\gamma_0 > 0$),
some block length $B$, and message length $b = (1-\gamma_0) B$.
We assume that $\cC_0$ is an instantiation of Construction~\ref{constr:prob}
and satisfies the properties promised by Lemma~\ref{lem:properties}.

\item
A coding scheme $\cC_1=(\enc_1, \dec_1)$ with
rate $r > 0$ (where $r$ can in general be sub-constant),
block length $n_1 := \gamma_1 n$ (where $n$ is defined later), 
and message length $k_1 := \gamma_1 r n$, that is non-malleable
against bit-tampering adversaries with error $\eps_1$.
Without loss of generality, assume that $\dec_1$ never
outputs $\perp$ (otherwise, identify $\perp$ with an arbitrary fixed message; e.g., $0^{k_1}$). 

\item
A linear error-correcting secret sharing
(LECSS) scheme $\cC_2=(\enc_2, \dec_2)$ (as in Definition~\ref{def:lecss}) 
with message length $k_2 := k$, rate $1-\gamma_2$
(for an arbitrarily small parameter $\gamma_2 > 0$) and block length $n_2$.
We assume that
$\cC_2$ is a $(\delta_2 n_2, t_2 := \gamma'_2 n_2)$-linear error-correcting secret sharing
scheme (where $\delta_2 > 0$ and $\gamma'_2 > 0$ are constants defined by
the choice of $\gamma_2$).
Since $b$ is a constant, without loss of generality assume that $b$ divides $n_2$, and let 
$n_b := n_2 / b$ and $n := n_2 B/b$.

\item A polynomial-time computable mapping $\Perm\colon \zo^{k_1} \to \cS_n$, where $\cS_n$ denotes the 
set of permutations on $[n]$. We assume that $\Perm(U_{k_1})$ is
an $\ell$-wise $\delta$-dependent permutation (as in Definition~\ref{def:limited:perm}, 
for parameters $\ell$ and $\delta$.
In fact, it is possible to achieve
$\delta \leq \exp(-\ell)$ and $\ell = \lceil \gamma_1 r n/\log n \rceil$ for some constant
$\gamma > 0$. Namely, we may use the following result due to 
Kaplan, Naor and Reingold \cite{ref:KNR05}:
\begin{thm}\cite{ref:KNR05} \label{thm:permutation}
For every integers $n, k_1 > 0$, there is a function
$\Perm\colon \zo^{k_1} \to \cS_n$ computable in
worst-case polynomial-time (in $k_1$ and $n$) such that
$\Perm(U_{k_1})$ is an $\ell$-wise $\delta$-dependent
permutation, where $\ell = \lceil k_1/\log n \rceil$ and
$\delta \leq \exp(-\ell)$.
\end{thm} \qed
\end{enumerate}

\subsubsection{The encoder}
\label{sec:encoder}

Let $s \in \zo^k$ be the message that we wish to encode.
The encoder generates the encoded message $\enc(s)$ according to the 
following procedure.

\begin{enumerate}
\item Let $Z \sim \cU_{k_1}$ and sample a random permutation $\Pi\colon [n] \to [n]$ 
by letting $\Pi := \Perm(Z)$. Let $Z' := \enc_1(Z) \in \zo^{\gamma_1 n}$.

\item Let $S' = \enc_2(s) \in \zo^{n_2}$ be the encoding of $s$ using the LECSS code $\cC_2$.

\item Partition $S'$ into blocks $S'_1, \ldots, S'_{n_b}$, each of length $b$, and encode each
block independently using $\cC_0$ so as to obtain a string
$C = (C_1, \ldots, C_{n_b}) \in \zo^{n}$. 

\item Let $C' := \Pi(C)$ be the string $C$ after its $n$ coordinates are
permuted by $\Pi$.

\item Output $\enc(s) := (Z', C') \in \zo^{N}$, where $N := (1+\gamma_1) n$, as the encoding of $s$.
\end{enumerate}

A schematic description of the encoder summarizing the involved parameters
is depicted in Figure~\ref{fig:explicit}.

\subsubsection{The decoder}
\label{sec:decoder}

We define the decoder $\dec(\bar{Z'}, \bar{C'})$ as follows:

\begin{enumerate}
\item Compute $\bar{Z} := \dec_1(\bar{Z'})$.

\item Compute the permutation $\bar{\Pi}\colon [n] \to [n]$ defined by $\bar{\Pi} := \Perm(\bar{Z})$.

\item Let $\bar{C} \in \zo^n$
be the permuted version of $\bar{C'}$ according to $\bar{\Pi}^{-1}$.

\item Partition $\bar{C}$ into $n_1/b$ blocks $\bar{C}_1, \ldots, \bar{C}_{n_b}$ 
of size $B$ each (consistent to the
way that the encoder does the partitioning of $\bar{C}$).

\item Call the inner code decoder on each block, namely, for each
$i \in [n_b]$ compute $\bar{S'}_i := \dec_0(\bar{C}_i)$. If $\bar{S'}_i = \perp$ for any $i$,
output $\perp$ and return.

\item Let $\bar{S'} = (\bar{S'}_1, \ldots, \bar{S'}_{n_b}) \in \zo^{n_2}$.
Compute $\bar{S} := \dec_2(\bar{S'})$, where $\bar{S} = \perp$ if $\bar{S'}$ is not a codeword of $\cC_2$. Output $\bar{S}$.
\end{enumerate}

\begin{remark} \label{rem:Justesen}
As in the classical variation of concatenated codes of Forney \cite{ref:forney} due to
Justesen \cite{ref:justesen}, the encoder described above can enumerate a \emph{family}
of inner codes instead of one fixed code
in order to eliminate the exhaustive search for a good inner code
$\cC_0$. In particular, one can consider all possible realizations of
Construction~\ref{constr:prob} for the chosen parameters and use each obtained
inner code to encode one of the $n_b$ inner blocks. If the fraction of
good inner codes (i.e., those satisfying the properties listed in Lemma~\ref{lem:properties}) 
is small enough (e.g., $1/n^{\Omega(1)}$), our analysis still
applies. It is possible to ensure that the size of the inner code family is
not larger than $n_b$ by appropriately choosing the parameter $\eta$ in Theorem~\ref{thm:upperBound}
(e.g., $\eta \geq 1/\sqrt{n}$).
\end{remark}

\subsection{Analysis}
\label{sec:analysis}

In this section, we prove that Construction of Section~\ref{sec:explicit:construction} is indeed
a coding scheme that is non-malleable against bit-tampering adversaries
with rate arbitrarily close to $1$.
More precisely, we prove the following theorem.

\begin{thm} \label{thm:explicit}
For every $\gamma_0 > 0$, there is a $\gamma'_0 = \gamma_0^{O(1)}$ 
and $N_0 = O(1/\gamma_0^{O(1)})$ such
that for every integer $N \geq N_0$, the following holds\footnote{
We can extend the construction to arbitrary block lengths $N$ by 
standard padding techniques and observing that the set of block lengths
for which construction of Figure~\ref{fig:explicit} is defined is
dense enough to allow padding without affecting the rate.
}.
The pair $(\enc, \dec)$ defined in Sections \ref{sec:encoder}~and~\ref{sec:decoder}
can be set up to be a  
non-malleable coding scheme 
against bit-tampering adversaries, achieving
block length $N$,
rate at least $1-\gamma_0$ and error 
\[
\eps \leq \eps_1 + 2 \exp\Big(-\Omega\Big(\frac{\gamma'_0 r N}{\log^3 N}\Big)\Big),
\]
where $r$ and $\eps_1$ are respectively the rate and the error of the assumed
non-malleable coding scheme $\cC_1$.
\end{thm}

\begin{remark}
Dziembowski et al.~\cite[Definition~3.3]{ref:nmc} also introduce a ``strong'' variation of non-malleable codes
which implies the standard definition (Definition~\ref{def:nmCode}) but is more restrictive.
It can be argued that the stronger definition is less natural in the sense that an
error-correcting code that is able to fully correct the tampering incurred by the adversary
does not satisfy the stronger definition while it is non-malleable in the standard sense,
which is what naturally expected to be the case.
In this work, we focus on the standard definition and prove the results with respect
to Definition~\ref{def:nmCode}. However, it can be verified (by minor adjustments of the
proof of Theorem~\ref{thm:explicit})
that the construction of
this section satisfies strong non-malleability (without any loss in the parameters) as well
provided that the non-malleable code $(\enc_1, \dec_1)$ encoding the description
of the permutation $\Pi$ satisfies the strong definition.
\end{remark}

\subsubsection*{Proof of Theorem~\ref{thm:explicit}}

It is clear that, given $(Z', C')$, the decoder can unambiguously reconstruct
the message $s$; that is, $\dec(\enc(s)) = s$ with probability $1$. 
Thus, it remains to demonstrate non-malleability of $\enc(s)$
against bit-tampering adversaries. 

Fix any such adversary $f\colon \zo^N \to \zo^N$.
The adversary $f$ defines the following partition of $[N]$:

\newcommand{\fr}{\mathsf{Fr}}
\newcommand{\fl}{\mathsf{Fl}}
\newcommand{\id}{\mathsf{Id}}

\begin{itemize}
\item $\fr \subseteq [N]$; the set of positions frozen to either zero or one by $f$.

\item $\fl \subseteq [N] \setminus \fr$; the set of positions flipped by $f$.

\item $\id = [N] \setminus (\fr \cup \fl)$; the set of positions left unchanged by $f$.
\end{itemize}

Since $f$ is not the identity function (otherwise, there is nothing to prove),
we know that $\fr \cup \fl \neq \emptyset$.


We use the notation used in the description of the encoder $\enc$ and decoder $\dec$ for
various random variables involved in the encoding and decoding of the message $s$.
In particular, let $(\bar{Z'}, \bar{C'}) = f(Z', C')$ denote the perturbation of $\enc(s)$ by the
adversary, and let $\bar{\Pi} := \Perm(\dec_1(\bar{Z'}))$ be the induced
perturbation of $\Pi$ as viewed by the decoder $\dec$. In general
$\Pi$ and $\bar{\Pi}$ are correlated random variables, but independent of
the remaining randomness used by the encoder.

We first distinguish three cases and subsequently use a convex combination argument to show that
the analysis of these cases
suffices to guarantee non-malleability in general. The first case
considers the situation where the adversary freezes too many bits of the encoding.
The remaining two cases can thus assume that a sizeable fraction of the bits
are not frozen to fixed values.

\subsection*{Case 1: Too many bits are frozen by the adversary.}

First, assume that $f$ freezes at least $n-t_2/b$ of the $n$ bits of
$C'$. In this case, we show that the distribution of 
$\dec(f(Z', C'))$ is always independent of the message $s$ and thus
the non-malleability condition of Definition~\ref{def:nmCode}
is satisfied for the chosen $f$. In order to achieve this goal, we
rely on bounded independence property of the LECSS code $\cC_2$.
We remark that a similar technique has been used in \cite{ref:nmc} for their
construction of non-malleable codes (and for the case where the adversary freezes
too many bits).

Observe that the joint distribution of $(\Pi, \bar{\Pi})$
is independent of the message $s$. Thus it suffices to show that 
conditioned on any realization $\Pi = \pi$ and $\bar{\Pi} = \bar{\pi}$, for any
fixed permutations $\pi$ and $\bar{\pi}$, the conditional distribution of 
$\dec(f(Z', C'))$ is independent of the message $s$.

We wish to understand how, with respect to the particular permutations
defined by $\pi$ and $\bar{\pi}$, the adversary
acts on the bits of the inner code blocks $C = (C_1, \ldots, C_{n_b})$. 

Consider the set $T \subseteq [n_b]$ of the 
blocks of $C=(C_1, \ldots, C_{n_b})$ (as defined in the algorithm for $\enc$)
that are not completely frozen by $f$ (after permuting the
action of $f$ with respect to the fixed choice of $\pi$). We know
that $|T| \leq t_2/b$.

Let $S'_T$ be the string $S' = (S'_1, \ldots, S'_{n_b})$ (as defined in the algorithm for $\enc$) restricted
to the blocks defined by $T$; that is,
$S'_T := (S'_i)_{i \in T}$. Observe that the length of $S'_T$ is at most
$b |T| \leq t_2$. From the $t_2$-wise independence property of the LECSS code
$\cC_2$, and the fact that the randomness of $\enc_2$ is independent of
$(\Pi, \bar{\Pi})$, we know that $S'_T$ is a uniform string, and in particular,
independent of the original message $s$. Let $C_T$ be the restriction
of $C$ to the blocks defined by $T$; that is,
$C_T := (C_i)_{i \in T}$. Since $C_T$ is generated from $S_T$ (by applying
the encoder $\enc_0$ on each block, whose randomness is independent of 
$(\Pi, \bar{\Pi})$),
we know that the distribution of $C_T$ is independent of the original
message $s$ as well.

Now, observe that $\dec(f(Z', C'))$ is only a function of $T$, $C_T$, the tampering
function $f$ and the fixed choices of $\pi$ and $\bar{\pi}$ 
(since the bits of $C$ that are not picked by $T$ are frozen 
to values determined by the tampering function $f$),
which are all independent of the message $s$. 
Thus in this case, $\dec(f(Z', C'))$
is independent of $s$ as well. This suffices to prove non-malleability of the code in this case.
In particular, in Definition~\ref{def:nmCode}, we can take $\cD_f$ to be
the distribution of $\dec(f(Z', C'))$ for an arbitrary message and satisfy the definition
with zero error.

\subsection*{Case 2: The adversary does not alter $\Pi$.}

In this case, we assume that $\Pi = \bar{\Pi}$, both distributed
according to $\Perm(\cU_{k_1})$ and independently of the remaining
randomness used by the encoder.
This situation in particular occurs if the adversary leaves the part of the encoding corresponding
to $Z'$ completely unchanged.
We furthermore assume that Case~1 does not occur;
i.e., more than $t_2/b = \gamma'_2 n_2/b$ bits of $C'$ are not frozen by the adversary.
To analyze this case, we rely
on bounded independence of the permutation $\Pi$. The effect of the
randomness of $\Pi$ is to prevent the adversary from gaining any advantage 
of the fact that the inner code independently acts on the individual blocks.

Let $\id' \subseteq \id$ be the positions of $C'$ that are left unchanged by $f$.
We know that $|\id' \cup \fl| > t_2/b$.
Moreover, the adversary freezes the bits of 
$C$ corresponding to the positions in $\Pi^{-1}(\fr)$ and either flips
or leaves the rest of the bits of $C$ unchanged.
We consider two sub-cases.

\subsubsection*{Case 2.1: $|\id'| > n - \delta_2 n_b$}
In this case, all but less than $\delta_2 n_b$ of the
inner code blocks are decoded to the correct values by the decoder.
Thus, the decoder correctly reconstructs all but less than
$b(n - |\id'|) \leq \delta_2 n_2$ bits of $S'$. Now, the distance property
of the LECSS code $\cC_2$ ensures that occurrence of any errors in $S'$
can be detected by the decoder.

Let $T_0 \subseteq [n_b]$ be the set of blocks of $C$ that are affected
by the action of $f$ (that is, those blocks in which there is a position
$i \in [n]$ where $\Pi(i) \notin \id$), and $T_1 \subseteq [n_2]$ (resp., 
$T_2 \subseteq [n]$) be the coordinate positions of $S'$ (resp., $C$)
contained in the blocks defined by $T_0$. Observe that $|T_0| < \delta_2 n_b$,
$|T_1| = b |T_0| < \delta_2 n_2$ and  $|T_2| = B |T_0|$.

The bounded independence property of $\cC_2$
ensures that the restriction of $S'$ to the positions in $T_1$
is uniformly distributed, provided that
\begin{equation} \label{eqn:explicit:assumption:gamma}
\gamma'_2 \geq \delta_2
\end{equation}
that we will assume in the sequel. Consequently, the restriction of $C$
to the positions in $T_2$ has the exact same distribution regardless
of the encoded message $s$. 

Recall that the decoder either outputs 
the correct message $s$ or $\perp$, and the former happens if
and only if $S'$ is correctly decoded at the positions in $T_1$.
This event (that is, $\bar{S}'|_{T_1} = S'|_{T_1}$) is independent of the encoded message $s$,
since the estimate $\bar{S}'|_{T_1}$ is completely determined by 
$S'|_{T_1}$, $\Pi$, and $f$, which are all independent of $s$. Thus,
the probability of the decoder outputting $\perp$ is the same
regardless of the message $s$. Since the decoder either outputs the
correct $s$ or $\perp$, we can conclude non-malleability of the code
in this case is achieved with zero error and a distribution
$\cD_f$ that is only supported on $\{ \same, \perp\}$.

\subsubsection*{Case 2.2: $|\id'| \leq n - \delta_2 n_b$}
In this case, we have $|\fr \cup \fl| \geq \delta_2 n_2/b$.
Moreover, we fix randomness of the LECSS $\cC_2$ so that $S'$ becomes a
fixed string. Recall that $C_1, \ldots, C_{n_b}$ are independent random variables,
since every call of the inner encoder $\enc_0$ uses fresh randomness.
In this case, our goal is to show that the decoder outputs $\perp$
with high probability, thus ensuring non-malleability by choosing $\cD_f$
to be the singleton distribution on $\{ \perp \}$.

Since $\Pi = \bar{\Pi}$, the decoder is able to correctly identify positions
of all the inner code blocks determined by $C$. In other words, we have
\[
\bar{C} = f'(C),
\]
where $f'$ denotes the adversary obtained from $f$ by permuting its
action on the bits as defined by $\Pi^{-1}$; that is,
\[
f'(x) := \Pi^{-1}(f(\Pi(x))).
\]

Let $i \in [n_b]$. We consider the dependence between $C_i$ and
its tampering $\bar{C}_i$, conditioned on the knowledge of $\Pi$
on the first $i-1$ blocks of $C$. Let $C(j)$ denote the $j$th
bit of $C$, so that the $i$th block of $C$ becomes
$(C(1+(i-1)B), \ldots, C(iB))$.  
For the moment, assume that $\delta = 0$; that is, $\Pi$ is exactly
a $\ell$-wise independent permutation. 

Suppose $i B \leq \ell$, meaning
that the restriction of $\Pi$ on the $i$th block (i.e., 
$(\Pi(1+(i-1)B), \ldots, \Pi(iB))$ conditioned on any fixing of
$(\Pi(1), \ldots, \Pi((i-1)B))$ exhibits the same distribution
as that of a uniformly random permutation.

We define events $\cE_1$ and $\cE_2$ as follows.
$\cE_1$ is the event that $\Pi(1+(i-1)B) \notin \id'$, 
and $\cE_2$ is the event that $\Pi(2+(i-1)B) \notin \fr$.
That is, $\cE_1$ occurs when the adversary does not
leave the first bit of the $i$th block of $C$ intact,
and $\cE_2$ occurs when the adversary does not freeze
the second bit of the $i$th block. We are interested
in lower bounding the probability that both
$\cE_1$ and $\cE_2$ occur, conditioned on any 
particular realization of $(\Pi(1), \ldots, \Pi((i-1)B))$.

Suppose the parameters are set up so that
\begin{equation} \label{eqn:explicit:assumption:a}
\ell \leq \frac{1}{2} \min\{ \delta_2 n_2/b, \gamma'_2 n_2/b \}.
\end{equation}

Under this assumption, even conditioned on any fixing of
$(\Pi(1), \ldots, \Pi((i-1)B))$, we can ensure that
\[
\Pr[\cE_1] \geq \delta_2 n_2/(2 b n),
\]
and
\[
\Pr[\cE_2 | \cE_1] \geq \gamma'_2 n_2/(2 b n),
\]
which together imply
\begin{equation} \label{eqn:explicit:gammaPP}
\Pr[\cE_1 \land \cE_2] \geq \delta_2 \gamma'_2 \Big(\frac{n_2}{2 b n}\Big)^2
=: \gamma''_2.
\end{equation}
We let $\gamma''_2$ to be the right hand side of the above inequality.

In general, when the random permutation is $\ell$-wise $\delta$-dependent 
for $\delta \geq 0$, the above lower bound can only
be affected by $\delta$. Thus, under the assumption that
\begin{equation} \label{eqn:explicit:assumption:b}
\delta \leq \gamma''_2/2,
\end{equation}
we may still ensure that 
\begin{equation} \label{eqn:explicit:slightChance}
\Pr[\cE_1 \land \cE_2] \geq \gamma''_2/2.
\end{equation}

Let $X_i \in \zo$ indicate the event that
$\dec_0(\bar{C}_i) = \perp$. We can write
\[
\Pr[X_i = 1] \geq \Pr[X_i = 1 | \cE_1 \land \cE_2] \Pr[\cE_1 \land \cE_2]
\geq (\gamma''_2/2) \Pr[X_i = 1 | \cE_1 \land \cE_2],
\]
where the last inequality follows from \eqref{eqn:explicit:slightChance}.
However, by property~\ref{prop:bottom} of Lemma~\ref{lem:properties}
that is attained by the inner code $\cC_0$, we also know that
\[
\Pr[X_i = 1 | \cE_1 \land \cE_2] \geq 1/3,
\]
and therefore it follows that
\begin{equation} \label{eqn:explicit:slightChance:b}
\Pr[X_i = 1] \geq \gamma''_2/6.
\end{equation}
Observe that by the argument above, \eqref{eqn:explicit:slightChance:b}
holds even conditioned on the realization of the permutation $\Pi$ 
on the first $i-1$ blocks of $C$. By recalling that we have
fixed the randomness of $\enc_2$, and that each inner block is
independently encoded by $\enc_0$, we can deduce that, letting
$X_0 := 0$,
\begin{equation} \label{eqn:explicit:slightChance:c}
\Pr[X_i = 1 | X_0, \ldots, X_{i-1}] \geq \gamma''_2/6.
\end{equation}
Using the above result for all $i \in \{1, \ldots, \lfloor \ell/B \rfloor \}$,
we conclude that
\begin{align}
\Pr[\dec(\bar{Z'}, \bar{C'}) \neq \perp] &\leq
\Pr[X_1 = X_2 = \cdots = X_{\lfloor \ell/B \rfloor} = 0] 
\label{eqn:explicit:slightChance:d} \\
&\leq \Big(1-\gamma''_2/6\Big)^{\lfloor \ell/B \rfloor},
\label{eqn:explicit:slightChance:e}
\end{align}
where \eqref{eqn:explicit:slightChance:d} holds since the
left hand side event is a subset of the right hand side event,
and \eqref{eqn:explicit:slightChance:e} follows from
\eqref{eqn:explicit:slightChance:c} and the chain rule.

Thus, by appropriately setting the parameters as we will do later,
we can ensure that the decoder outputs $\perp$ with high probability.
This ensures non-malleability of the code in this case with the choice of
$\cD_f$ in Definition~\ref{def:nmCode} being entirely supported on $\{ \perp \}$ and error
bounded by \eqref{eqn:explicit:slightChance:e}.

\subsection*{Case 3: The decoder estimates an independent permutation.}

In this case, we consider the event that $\bar{\Pi}$ attains a particular
value $\bar{\pi}$. Suppose it so happens that under this conditioning, the distribution
of $\Pi$ remains unaffected; that is, $\bar{\Pi} = \pi$ and $\Pi \sim \Perm(\cU_{k_1})$. 
This situation may occur if the adversary completely freezes the part of the encoding corresponding
to $Z'$ to a fixed valid codeword of $\cC_1$.
Recall that
the random variable $\Pi$ is determined by the random string $Z$ and that it is
independent of the remaining randomness used by the encoder $\enc$.
Similar to the previous case, 
our goal is to upper bound the probability that $\dec$ does not output $\perp$.
Furthermore, we can again assume that Case~1 does not occur;
i.e., more than $t_2/b$ bits of $C'$ are not frozen by the adversary.
For the analysis of this case, we can fix the randomness of $\enc_2$ and 
thus assume that $S'$ is fixed to a particular value.

As before, our goal is to determine how each block $C_i$ of the inner code
is related to its perturbation $\bar{C}_i$ induced by the adversary.
Recall that
\[
\bar{C} = \bar{\pi}^{-1}(f(\Pi(C))).
\]
Since $f$ is fixed to an arbitrary choice only with restrictions on the
number of frozen bits, without loss of generality we can assume that
$\bar{\pi}$ is the identity permutation (if not, permute the action of
$f$ accordingly), and therefore, $\bar{C'} = \bar{C}$ (since $\bar{C'} = \bar{\pi}(\bar{C})$), and
\[
\bar{C} = f(\Pi(C)).
\]
For any $\tau \in [n_b]$, let $f_\tau\colon \zo^B \to \zo^B$ denote the restriction of the
adversary to the positions included in the $\tau$th block of $\bar{C}$.
%

Assuming that 
$\ell \leq t_2$ (which is implied by \eqref{eqn:explicit:assumption:a}),
let $T \subseteq [n]$ be any set of 
size $\lfloor \ell/B \rfloor \leq \lfloor t_2/B \rfloor \leq t_2/b$ of the
coordinate positions of $C'$ that are either left unchanged or flipped
by $f$. Let $T' \subseteq [n_b]$ (where $|T'| \leq |T|$) be the 
set of blocks of $\bar{C}$ that contain the positions picked by $T$.
With slight abuse of notation, for any $\tau \in T'$, denote by $\Pi^{-1}(\tau) \subseteq [n]$
the set of indices of the positions belonging to the block $\tau$ 
after applying the permutation $\Pi^{-1}$ to each one of them. In other words,
$\bar{C}_{\tau}$ (the $\tau$th block of $\bar{C}$) is determined by
taking the restriction of $C$ to the bits in $\Pi^{-1}(\tau)$ (in their
respective order), and applying
$f_\tau$ on those bits (recall that for $\tau \in T'$ we are guaranteed that $f_\tau$
does not freeze all the bits).

In the sequel, our goal is to show that with high probability, $\dec(\bar{Z}, \bar{C'}) = \perp$.
In order to do so, we first assume that $\delta = 0$; i.e., that $\Pi$ is exactly
an $\ell$-wise independent permutation. 
Suppose $T' = \{ \tau_1, \ldots, \tau_{|T'|} \}$, and consider any $i \in |T'|$.

We wish to lower bound the probability that $\dec_0(\bar{C}_{\tau_i}) = \perp$,
conditioned on the knowledge of $\Pi$ on the first $i-1$ blocks in $T'$.
Subject to the conditioning, the values of $\Pi$ becomes known on up to
$(i-1)B \leq (|T'|-1)B \leq \ell-B$ points. Since $\Pi$ is $\ell$-wise independent,
$\Pi$ on the $B$ bits belonging to the $i$th block remains $B$-wise
independent. Now, assuming
\begin{equation} \label{eqn:assumption:ell}
\ell \leq n/2,
\end{equation}
we know that even subject to the knowledge of $\Pi$ on any $\ell$ positions of $C$, 
the probability that a uniformly random element within the remaining positions
falls in a particular block of $C$ is at most $B/(n-\ell) \leq 2B/n$.

Now, for $j \in \{2, \ldots, B\}$, consider the $j$th position of the block $\tau_i$ in $T'$.
By the above argument, the probability that $\Pi^{-1}$ maps this element to a block of $C$ chosen by
any of the previous $j-1$ elements is at most $2B/n$. By a union bound on the choices of $j$,
with probability at least
\[
1-2B^2/n,
\]
the elements of the block $\tau_i$ all land in distinct blocks of $C$ by the
permutation $\Pi^{-1}$. Now we observe that if $\delta > 0$, the above probability
is only affected by at most $\delta$.
Moreover, if the above distinctness property occurs, the values of $C$
at the positions in $\Pi^{-1}(\tau)$ become independent random bits; since
$\enc$ uses fresh randomness upon each call of $\enc_0$ for encoding different
blocks of the inner code (recall that the randomness of the first layer using $\enc_2$
is fixed). 

Recall that by the bounded independence property of $\cC_0$ (i.e., property~\ref{prop:indep} 
of Lemma~\ref{lem:properties}), each individual bit of $C$ is $\exp(-\Omega(\gamma_0 B))$-close to
uniform. Therefore, using
Proposition~\ref{prop:pBiased},
with probability at least $1-2B^2/n-\delta$ (in particular, at least $7/8$ when
\begin{equation} \label{eqn:assumption:N}
n \geq 32 B^2
\end{equation}
and assuming $\delta \leq 1/16$) 
we can ensure that the distribution of $C$ restricted to positions
picked by $\Pi^{-1}(\tau)$ is $O(B \exp(-\Omega(\gamma_0 B)))$-close to uniform,
or in particular $(1/4)$-close to uniform when $B$ is larger than a suitable constant.
If this happens, we can conclude that distribution of the block $\tau_i$ of $\bar{C}$ 
is $(1/4)$-close to a sub-cube with at least one random bit (since we have
assumed that $\tau \in T'$ and thus $f$ does not fix all the bit of the $\tau$th block). 
Now, the cube property of $\cC_0$ (i.e., property~\ref{prop:cube} of Lemma~\ref{lem:properties})
implies that
\[
\Pr_{\enc_0}[\dec_0(\bar{C}_{\tau_i}) \neq \perp | \Pi(\tau_1), \ldots, \Pi(\tau_{i-1})] \leq 1/2 + 1/4 = 3/4,
\]
where the extra term $1/4$ accounts for the statistical distance of
$\bar{C}_{\tau_i}$ from being a perfect sub-cube.

Finally, using the above probability bound, and running $i$ over all the blocks in $T'$,
and recalling the assumption that $\bar{C} = \bar{C'}$, we deduce that
\begin{equation}
\label{eqn:explicit:indep:failure}
\Pr[\dec(\bar{Z'}, \bar{C'}) \neq \perp] \leq (7/8)^{|T'|} \leq \exp(-\Omega(\ell/B^2)),
\end{equation}
where the last inequality follows from the fact that
$|T'| \geq \lfloor \ell/b \rfloor /B$.

In a similar way to Case~2.2 above, this concludes 
non-malleability of the code in this case with the choice of
$\cD_f$ in Definition~\ref{def:nmCode} being entirely supported on $\{ \perp \}$ and error
bounded by the right hand side of \eqref{eqn:explicit:indep:failure}.

\subsection*{The general case.}
Recall that Case~1 eliminates the situation in which the adversary freezes too many
of the bits. For the remaining cases, Cases 2~and~3 consider the special situations
where the two permutations $\Pi$ and $\bar{\Pi}$ used by the encoder and
the decoder either completely match or are completely independent. 
However, in general we may not reach any of the two cases. Fortunately,
the fact that the code $\cC_1$ encoding the permutation $\Pi$ is non-malleable
ensures that we always end up with a \emph{combination} of the Case~2 and 3.
In other words, in order to analyze any event depending on the joint distribution
of $(\Pi, \bar{\Pi})$, it suffices to consider the two special cases where
$\Pi$ is always the same as $\bar{\Pi}$, or when $\Pi$ and $\bar{\Pi}$ are
fully independent.

The joint distribution of $(\Pi, \bar{\Pi})$ may be understood using 
Lemma~\ref{lem:nmc:joint}. Namely, the lemma applied on the non-malleable
code $\cC_1$ implies that the joint distribution of $(\Pi, \bar{\Pi})$ is
$\eps_1$-close (recall that $\eps_1$ is the error of non-malleable code $\cC_1$) to the convex combination
\begin{equation} \label{eqn:c1:convex}
\alpha \cdot \distr(\Pi, \Pi) + (1-\alpha) \cdot \distr(\Pi, \Pi'), 
\end{equation}
for some parameter $\alpha \in [0,1]$ and an independent random variable $\Pi'$ distributed over $\cS_n$.

For a random variable $\bar{P}$ jointly distributed with $\Pi$ over $\cS_n$,
and with a slight overload of notation, define the random variable $D_{s,\bar{P}}$ 
over $\zo^k \cup \{\perp\}$ as the output of the following experiment
(recall that $s \in \zo^k$ is the message to be encoded):

\begin{enumerate}
\item Let $(Z', C') := \enc(s)$ be the encoding $\enc(s)$, as described in Section~\ref{sec:encoder}),
and $(\bar{Z}', \bar{C}') = f(Z', C')$ be the corrupted codeword under the adversary $f$.

\item Apply the decoder's procedure, described in Section~\ref{sec:decoder}, where in the second
line of the procedure the assignment $\bar{\Pi} := \Perm(\bar{Z})$ is replaced with $\bar{\Pi} := \bar{P}$,
and output the result.
\end{enumerate}

Intuitively, $D_{s, \bar{P}}$ captures decoding of the perturbed codeword when
the decoder uses an arbitrary estimate $\bar{P}$ (given in the subscript) of the random permutation $\Pi$
instead of reading it off the codeword (i.e., instead of $\bar{\Pi} := \Perm(\bar{Z})$ defined
by the decoder's procedure). Using this notation, $D_{s, \bar{\Pi}}$ (that is, when the choice of $\bar{P}$
is indeed the natural estimate $\Perm(\bar{Z})$) is the same as $\dec(f(\enc(s)))$. 

Define $D_s := D_{s, \bar{\Pi}}$, $D'_s := D_{s, \Pi}$ and $D''_s := D_{s, \Pi'}$. 
The results obtained in Cases 2~and~3 can be summarized as follows:
\begin{itemize}
\item (Case~2): There is a distribution $\cD'_f$ over $\zo^k \cup \{ \same, \perp \}$ such that the statistical distance between the 
distribution of $D'_s$ and $\Copy(\cD'_f, s)$ (that is, the distribution obtained by
reassigning the mass of $\same$ in $\cD'_f$ to $s$) is at most
\[
\Big(1-\gamma''_2/6\Big)^{\lfloor \ell/B \rfloor}
\]
(in fact, $\cD'_s$ is entirely supported on $\{ \perp, \same \}$).
 
\item (Case~3): There is a distribution $\cD''_f$ over $\zo^k \cup \{ \same, \perp \}$ such that the statistical distance between the 
distribution of $D''_s$ and $\Copy(\cD''_f, s)$ is at most
\[
\exp(-\Omega(\ell/B^2))
\]
(in fact, $\cD''_s$ is the distribution entirely supported on $\{ \perp \}$).
 
\end{itemize}

The convex decomposition \eqref{eqn:c1:convex} implies that the distribution of $D_s$
may be decomposed as a convex combination as well, that is (recalling that
the action of $f$ on the part of the codeword corresponding to $C'$ is independent
of $Z'$ and that the randomness used by $\enc_1$ is independent of the
randomness used by $\enc_2$ and each invocation of $\enc_0$),

\begin{equation} \label{eqn:ds:conex}
\distr(D_s) \approx_{\eps_1} \alpha \distr(D'_s) + (1-\alpha) \distr(D''_s).
\end{equation}
Now we set 
\[
\cD_f := \alpha \cD'_f + (1-\alpha) \cD''_f
\]
and define
\begin{equation}
\label{eqn:explicit:errorBound}
{\eps'} := \Big(1-\gamma''_2/6\Big)^{\lfloor \ell/B \rfloor} +
\exp(-\Omega(\ell/B^2)) + \eps_1.
\end{equation}

From the above observations, it follows that the distribution of $D_s$ (equivalently,
$\dec(f(\enc(s)))$) is ${\eps'}$-close to $\Copy(\cD_f, s)$. This proves 
non-malleability of the code in the general case with error bounded by ${\eps'}$.

\subsection*{Setting up the parameters} 

The final encoder $\enc$ maps $k$ bits into 
\[
\Big(\frac{k}{1-\gamma_2} \cdot \frac{1}{1-\gamma_0}\Big) (1+\gamma_1)
\]
bits. Thus the rate $r$ of the final code is
\[
r = \frac{(1-\gamma_0)(1-\gamma_2)}{1+\gamma_1}.
\]
We set up $\gamma_1, \gamma_2 \in [\gamma_0/2, \gamma_0]$ so as to
ensure that 
\[
r \geq 1-O(\gamma_0).
\]
Thus, the rate of the final code can be made arbitrarily close to $1$
if $\gamma_0$ is chosen to be a sufficiently small constant.

Before proceeding with the choice of other parameters, we recap the constraints
that we have assumed on the parameters; namely, 
\eqref{eqn:explicit:assumption:gamma},
\eqref{eqn:assumption:N:recap},
\eqref{eqn:explicit:assumption:a:recap}, 
\eqref{eqn:assumption:ell}, 
\eqref{eqn:explicit:assumption:b} (where we recall that 
$\gamma''_2 = \delta_2 \gamma'_2 (\frac{n_2}{2 b n})^2$)
which are again listed below to assist the reader.

\begin{gather} 
\gamma'_2 \geq \delta_2 \label{eqn:explicit:assumption:gamma:recap} \\
n \geq 32 B^2, \label{eqn:assumption:N:recap} \\
\ell \leq \frac{1}{2} \min\{ \delta_2 n_2/b, \gamma'_2 n_2/b \}, \label{eqn:explicit:assumption:a:recap} \\
\ell \leq n/2, \label{eqn:assumption:ell:recap} \\
\delta \leq \gamma''_2/2, \label{eqn:explicit:assumption:b:recap} 
\end{gather}

For the particular choice of $\gamma_0$, there is a constant 
\begin{equation} \label{eqn:explicit:B}
B=O((\log^2 \gamma_0)/\gamma_0)
\end{equation} for which
Lemma~\ref{lem:properties} holds. 

Note that the choice of $B$ only depends on the constant $\gamma_0$.
If desired, a brute-force search\footnote{Alternatively, it is possible to
sample a random choice for $\cC_0$ and then verify that it satisfies properties
of Lemma~\ref{lem:properties}, thereby obtaining a Las Vegas construction
which is more efficient (in terms of the dependence on the constant $\gamma_0$) than a 
brute-force search. The construction would be
even more efficient in Monte Carlo form; i.e., if one avoids verification of 
the candidate $\cC_0$.
} can thus find an explicit choice for the inner
code $\cC_0$ in time only depending on $\gamma_0$. Moreover, \eqref{eqn:assumption:N:recap}
can be satisfied as long as $N \geq N_0$ for some $N_0 = \poly(1/\gamma_0)$.

Now, for the assumed value for the constant $\gamma_2 \approx \gamma_0$, one
can use Corollary~\ref{coro:lecss} and set up 
$\cC_2$ to be an 
$(\Omega(\gamma_0 n_2/\log n_2), \Omega(\gamma_0 n_2/\log n_2))$-linear error-correcting secret sharing code.
Thus, we may assume that
$\delta_2 = \gamma'_2 = \Omega(\gamma_0 / \log N)$ (since, trivially,
$n_2 \leq N$) and also satisfy \eqref{eqn:explicit:assumption:gamma:recap}.

Finally, using Theorem~\ref{thm:permutation} we can set up 
$\Perm$ so that $\ell = \Omega(\gamma_1 r n/\log n)
= \Omega(\gamma_0 r n/\log n)$ and
$\delta \leq 1/n^\ell$. 
We can lower the value of $\ell$ if necessary (since an $\ell$-wise
$\delta$-dependent permutation is also an $\ell'$-wise
$\delta$-dependent permutation for any $\ell' \leq \ell$)
so as to ensure that
$\ell = \Omega(\gamma_0 r n/(B \log n))$ and
the assumptions \eqref{eqn:explicit:assumption:a:recap}~and~
\eqref{eqn:assumption:ell:recap} are satisfied
(recall that $n_2/b = n_b = n/B$ and $r \leq 1$).
Observe that our choices of the parameters implies that
the quantity $\gamma''_2$ defined in \eqref{eqn:explicit:gammaPP}
satisfies $\gamma''_2 = \Omega(\gamma_0^2 / (B \log N)^2)$.
We see that the choice of $\delta$ is small enough to satisfy
the assumption~\eqref{eqn:explicit:assumption:b:recap}.

By our choice of the parameters, the upper bound on the failure
probability in \eqref{eqn:explicit:slightChance:e} is
\begin{equation} \label{eqn:explicit:failure:a}
\Big(1-\gamma''_2/6\Big)^{\lfloor \ell/B \rfloor} =
\exp\Big(-\Omega\Big(\frac{\gamma_0^3 r N}{B^3 \log^3 N}\Big)\Big),
\end{equation}
which can be seen by recalling the lower bound on 
$\gamma''_2$ and the fact that $N=n(1+\gamma_1) \in [n, 2n]$.

On the other hand, the upper bound on the failure probability
in \eqref{eqn:explicit:indep:failure} can be written as
\begin{equation} \label{eqn:explicit:failure:b}
 \exp(-\Omega(\ell/B^2)) =
 \exp\Big(-\Omega\Big(\frac{\gamma_0 r N}{B^3 \log N}\Big)\Big),
\end{equation}
which is dominated by the estimate in \eqref{eqn:explicit:failure:a}.


Now we can substitute the upper bound \eqref{eqn:explicit:B} on $B$
to conclude that \eqref{eqn:explicit:failure:a} is at most
\[
\exp\Big(-\Omega\Big(\frac{\gamma_0^6 r N}{\log^6(1/\gamma_0) \log^3 N}\Big)\Big)
= \exp\Big(-\Omega\Big(\frac{\gamma'_0 r N}{\log^3 N}\Big)\Big),
\]
where
\[
\gamma'_0 := (\gamma_0 / \log(1/\gamma_0))^6.
\]

We conclude that the error of the final coding scheme $(\enc, \dec)$
which is upper bounded by ${\eps'}$ as defined in \eqref{eqn:explicit:errorBound}
is at most
\[
\eps_1 + 2 \exp\Big(-\Omega\Big(\frac{\gamma'_0 r N}{\log^3 N}\Big)\Big).
\]

\subsection{Instantiations} \label{sec:instantiations}

We present two possible choices for the non-malleable code $\cC_1$
based on existing constructions.
The first construction, due to Dziembowski et al.~\cite{ref:nmc},
is a Monte Carlo result that is summarized below.

\begin{thm} \cite[Theorem~4.2]{ref:nmc} \label{thm:DPW}
For every integer $n > 0$, there is an efficient coding scheme
$\cC_1$ of block length $n$, rate at least $.18$, that is non-malleable
against bit-tampering adversaries achieving error $\eps = \exp(-\Omega(n))$.
Moreover, there is an efficient randomized algorithm that,
given $n$, outputs a description of such a code with probability
at least $1-\exp(-\Omega(n))$. 
\end{thm}

More recently, Aggarwal et al.~\cite{ref:ADL} construct an \emph{explicit}
coding scheme which is non-malleable against the much more general class
of split-state adversaries. However, this construction achieves
inferior guarantees than the one above in terms of the rate and error.
Below we rephrase this result restricted to bit-tampering adversaries.

\begin{thm} \cite[implied by Theorem~5]{ref:ADL} \label{thm:ADL}
For every integer $k > 0$ and $\eps > 0$, there is an efficient 
and explicit\footnote{To be precise, explicitness is guaranteed
assuming that a large prime $p = \exp(\tilde{\Omega}(k+\log(1/\eps)))$
is available.} coding scheme
$\cC_1$ of message length $k$ that is non-malleable
against bit-tampering adversaries achieving error at most $\eps$.
Moreover, the block length $n$ of the coding scheme satisfies
\[
n = \tilde{O}((k+\log(1/\eps))^7).
\]
By choosing $\eps := \exp(-k)$, we see that
we can have $\eps = \exp(-\tilde{\Omega}(n^{1/7}))$ while 
the rate $r$ of the code satisfies
\[
r = \tilde{\Omega}(n^{-6/7}).
\]
\end{thm}

By instantiating Theorem~\ref{thm:explicit} with the Monte Carlo
construction of Theorem~\ref{thm:DPW}, we arrive at the following
corollary.

\begin{coro} \label{coro:explicit:DPW}
For every integer $n > 0$ and
every positive parameter $\gamma_0 = \Omega(1/(\log n)^{O(1)})$, 
there is an efficient
coding scheme $(\enc, \dec)$ of block length $n$
and rate at least $1-\gamma_0$ such that the following hold.
\begin{enumerate}
\item The coding scheme is 
non-malleable against bit-tampering 
adversaries, achieving error at most $\exp(-\tilde{\Omega}(n))$,

\item There is an efficient randomized algorithm that,
given $n$, outputs a description of such a code with probability
at least $1-\exp(-\Omega(n))$. 
\end{enumerate} \qed
\end{coro}

If, instead, we instantiate Theorem~\ref{thm:explicit} with the 
construction of Theorem~\ref{thm:ADL}, we obtain the following
non-malleable code.

\begin{coro} \label{coro:explicit:ADL}
For every integer $n > 0$ and
every positive parameter $\gamma_0 = \Omega(1/(\log n)^{O(1)})$, 
there is an explicit and efficient
coding scheme $(\enc, \dec)$ of block length $n$
and rate at least $1-\gamma_0$ such that the 
coding scheme is 
non-malleable against bit-tampering 
adversaries and achieves error at most 
$\exp(-\tilde{\Omega}(n^{1/7}))$.
\qed
\end{coro}

\section{Construction of non-malleable codes using non-malleable extractors}
\label{sec:nmext}

In this section, we introduce the notion of seedless non-malleable extractors
that extends the existing definition of seeded non-malleable
extractors (as defined in \cite{ref:DW09}) to sources that exhibit structures of interest.
This is similar to how classical seedless extractors
are defined as an extension of seeded extractors to sources
with different kinds of structure\footnote{For a background on
standard seeded and seedless extractors, see \cite[Chapter~2]{ref:Che10}.}.

Furthermore, we obtain a reduction from the non-malleable variation of two-source
extractors to non-malleable codes for the split-state model. 
Dziembowski et al.~\cite{ref:DKO} obtain a construction of non-malleable codes
encoding one-bit messages based on a variation of strong (standard) two-source extractors.
This brings up the question of whether there is a natural variation of
two-source extractors that directly leads to non-malleable codes for the split-state
model encoding messages of arbitrary lengths (and ideally, achieving constant rate). Our notion
of non-malleable two-source extractors can be regarded as a positive answer to this question.

Our reduction does not imply a characterization of non-malleable codes using
extractors, and non-malleable codes for the split-state model do not necessarily
correspond to non-malleable extractors (since those implied by our reduction achieve
slightly sub-optimal rates). However, since seeded
non-malleable extractors (as studied in the line of research starting \cite{ref:DW09}) are already subject
of independent interest, we believe our characterization may be seen as 
a natural approach (albeit not the only possible approach) for improved constructions
of non-malleable codes. 
Furthermore, the definition of two-source non-malleable extractors (especially
the criteria described in Remark~\ref{rem:nmext:relaxed} below) is somewhat cleaner
and easier to work with than then definition of non-malleable codes (Definition~\ref{def:nmCode})
that involves subtleties such as the extra care for the ``$\same$'' symbol.

It should also be noted that our reduction can be
modified to obtain non-malleable codes for different classes of adversaries (by appropriately
defining the family of extractors based on the tampering family being considered)
such as the variation of split-state model where the adversary may arbitrarily choose in advance
how to partition the encoding into two blocks (in which case one has to consider
the non-malleable variation of mixed two-source extractors studied by Raz
and Yehudayoff \cite{ref:RY11}).

\subsection{Seedless non-malleable extractors}
\label{sec:nmext:seedless}

Before defining seedless non-malleable extractors, it is
convenient to introduce a related notion of \emph{non-malleable 
functions} that is defined with respect to a function and
a distribution over its inputs.
As it turns out, non-malleable ``extractor'' functions with respect to the uniform
distribution and limited families of adversaries are of particular interest
for construction of non-malleable codes.

\begin{defn} \label{def:nmfunc}
A function $g\colon \Sigma \to \Gamma$ is said to be non-malleable with error $\eps$ with
respect to a distribution $\cX$ over $\Sigma$ and a tampering function
$f\colon \Sigma \to \Sigma$  if there is a distribution
$\cD$ over $\Gamma \cup \{\same\}$ such that for an independent
$Y \sim \cD$, 
\[
\distr(g(X), g(f(X))) \approx_\eps \distr(g(X), \Copy(Y, g(X))).
\]
\end{defn}

Using the above notation, we can now define seedless non-malleable 
extractors as follows.

\begin{defn} \label{def:nmext:seedless}
A function $\nm\colon \zo^n \to \zo^m$ is a (seedless) non-malleable
extractor with respect to a class $\mathfrak{X}$ of sources over
$\zo^n$ and a class $\cF$ of tampering functions acting on $\zo^n$
if, for every distribution $\cX \in \mathfrak{X}$,
and for every tampering function $f \in \cF$, $f\colon \zo^n \to \zo^n$,
the following hold for an error parameter $\eps > 0$.
\begin{enumerate}
\item $\nm$ is an extractor for the distribution $\cX$; that is,
$
\nm(\cX) \approx_\eps \U_m.
$

\item $\nm$ is a non-malleable function with error $\eps$ for the distribution
$\cX$ and with respect to the tampering function $f$. 
\end{enumerate}
\end{defn}

Of particular interest is the notion of \emph{two-source} seedless
extractors. This is a special case of Definition~\ref{def:nmext:seedless}
where $\mathfrak{X}$ is the family of two-sources (i.e.,
each $\cX$ is a product distribution $(\cX_1, \cX_2)$, where
$\cX_1$ and $\cX_2$ are arbitrary distributions 
defined over the first and second half of the input, each
having a sufficient amount of entropy. Moreover, the family
of tampering functions consists of functions that arbitrary
but independently tamper each half of the input.
Formally, we distinguish
this special case of Definition~\ref{def:nmext:seedless} as follows.

\begin{defn} \label{def:nmext:strict}
A function $\nm\colon \zo^n \times \zo^n \to \zo^m$ is a two-source non-malleable
$(k_1, k_2, \eps)$-extractor if, for every product distribution $(\cX, \cY)$ over
$\zo^n \times \zo^n$ where $\cX$ and $\cY$ have min-entropy at least $k_1$ and $k_2$,
respectively, and for any arbitrary functions $f_1\colon \zo^n \to \zo^n$ and
$f_2\colon \zo^n \to \zo^n$, the following hold.
\begin{enumerate}
\item $\nm$ is a two-source extractor for $(\cX, \cY)$; that is,
$
\nm(\cX, \cY) \approx_\eps \U_m.
$

\item $\nm$ is a non-malleable function with error $\eps$ for the distribution
$(\cX, \cY)$ and with respect to the tampering function 
$
(X, Y) \mapsto (f_1(X), f_2(Y))
$.
\end{enumerate}
\end{defn}

In general, a tampering function may have fixed points and act
as the identity function on a particular set of inputs. Definitions
of non-malleable codes, functions, and extractors all handle
the technicalities involved with such fixed points by introducing
a special symbol $\same$. Nevertheless, it is more convenient to deal
with adversaries that are promised to have no fixed points.
For this restricted model, the definition of two-source 
non-malleable extractors can be modified as follows. We call
extractors satisfying the less stringent requirement
\emph{relaxed} two-source non-malleable extractors. Formally,
the relaxed definition is as follows.

\begin{defn} \label{def:nmext:relaxed}
A function $\nm\colon \zo^n \times \zo^n \to \zo^m$ is a relaxed two-source non-malleable
$(k_1, k_2, \eps)$-extractor if, for every product distribution $(\cX, \cY)$ over
$\zo^n \times \zo^n$ where $\cX$ and $\cY$ have min-entropy at least $k_1$ and $k_2$,
respectively, the following holds. Let $f_1\colon \zo^n \times \zo^n$ and
$f_2\colon \zo^n \times \zo^n$ be functions such that for every $x \in \zo^n$,
$f_1(x) \neq x$ and $f_2(x) \neq x$. Then, for $(X, Y) \sim (\cX, \cY)$,
\begin{enumerate}
\item \label{def:nmext:relaxed:nm}
$\nm$ is a two-source extractor for $(\cX, \cY)$; that is,
$
\nm(\cX, \cY) \approx_\eps \U_m.
$
\item $\nm$ is a non-malleable function with error $\eps$ for the distribution
of $(X, Y)$ and with respect to all of the tampering functions
\begin{align*}
(X, Y) \mapsto (f_1(X), Y), \qquad
(X, Y) \mapsto (X, f_2(Y)), \qquad
(X, Y) \mapsto (f_1(X), f_2(Y)). 
\end{align*}
\end{enumerate}
\end{defn}

\begin{remark} \label{rem:nmext:relaxed}
In order to satisfy the requirements of
Definition~\ref{def:nmext:relaxed}, it 
suffices (but not necessary) to ensure that
\begin{align*}
(\nm(\cX, \cY), \nm(f_1(\cX), \cY)) &\approx_\eps (U_m, \nm(f_1(\cX), \cY)), \\
(\nm(\cX, \cY), \nm(\cX, f_2(\cY))) &\approx_\eps (U_m, \nm(\cX, f_2(\cY))), \\
(\nm(\cX, \cY), \nm(f_1(\cX), f_2(\cY))) &\approx_\eps (U_m, \nm(f_1(\cX), f_2(\cY))).
\end{align*}
The proof of Theorem~\ref{thm:nmext:tool} shows that these stronger
requirements can be satisfied with high probability by random functions.
\end{remark}

It immediately follows from the definitions that a two-source non-malleable
extractor (according to Definition~\ref{def:nmext:strict}) is a 
relaxed non-malleable two-source extractor (according to
Definition~\ref{def:nmext:relaxed}) and with the same parameters. 
However, non-malleable extractors are in general
meaningful for arbitrary tampering functions that may potentially have fixed points.
Interestingly, below we show that the two notions are equivalent
up to a slight loss in the parameters.

\begin{lem} \label{lem:nmext:nofixedpoint}
Let $\nm$ be a relaxed two-source non-malleable $(k_1 - \log(1/\eps), k_2 - \log(1/\eps), \eps)$-extractor. Then,
$\nm$ is a two-source non-malleable $(k_1, k_2, 4 \eps)$-extractor. 
\end{lem}
\begin{proof}
Since the two-source extraction requirement of Definition~\ref{def:nmext:relaxed}
implies the extraction requirement of Definition~\ref{def:nmext:strict},
it suffices to prove the non-malleability condition of Definition~\ref{def:nmext:strict}.

Let $f_1\colon \zo^n \to \zo^n$ and $f_2\colon \zo^n \to \zo^n$ be a pair of
tampering functions, $(\cX, \cY)$ be a product (without loss of generality, component-wise flat) distribution with min-entropy
at least $(k_1, k_2)$, and $(X, Y) \sim (\cX, \cY)$. 
Define the parameters
\begin{gather*}
\eps_1 := \Pr[f_1(X) = X], \\
\eps_2 := \Pr[f_2(Y) = Y].
\end{gather*}
Moreover, define the distributions $\cX_0, \cX_1$ to be the 
distribution of $X$ conditioned on the events $f_1(X) = X$ and 
$f_1(X) \neq X$, respectively. Let $\cY_0, \cY_1$ be similar conditional
distributions for the random variable $Y$ and the events
$f_2(Y) = Y$ and $f_2(Y) \neq Y$. Let $X_0, X_1, Y_0, Y_1$ be random 
variables drawn independently and in order from $\cX_0, \cX_1, \cY_0, \cY_1$.
Observe that $(\cX, \cY)$ is now a convex combination of four
product distributions:
\[
(\cX, \cY) = \alpha_{00} (\cX_0, \cY_0) + \alpha_{01}(\cX_0, \cY_1) + 
\alpha_{10}(\cX_1, \cY_0) + \alpha_{11}(\cX_1, \cY_1)
\]
where 
\begin{align*}
\alpha_{00} &:= \eps_1 \eps_2, \\
\alpha_{01} &:= \eps_1 (1-\eps_2), \\
\alpha_{10} &:= (1-\eps_1) \eps_2, \\
\alpha_{11} &:= (1-\eps_1) (1-\eps_2).
\end{align*}
We now need to verify Definition~\ref{def:nmext:strict} for the tampering function
\begin{gather*}
(X, Y) \mapsto (f_1(X), f_2(Y)).
\end{gather*}
Let us consider the distribution \[\cE_{01} := \nm(f_1(\cX_0), f_2(\cY_1))
= \nm(\cX_0, f_2(\cY_1)).\]
Suppose $\alpha_{01} \geq \eps$, which implies $\eps_1 \geq \eps$ and $1-\eps_2 \geq \eps$.
Thus, $\cX_0$ and $\cY_1$ have min-entropy at least $k_1 - \log(1/\eps)$
and $k_2 - \log(1/\eps)$, respectively. In particular, since $f_2(\cY_1)$
has no fixed points, by Definitions \ref{def:nmext:relaxed} and \ref{def:nmfunc},
there is an distribution $\cD_{01}$ over $\zo^m \cup \{ \same \}$ (where $m$
is the output length of $\nm$) such that for an independent random
variable $E_{01} \sim \cD_{01}$, 
\[
\distr(\nm(X_0, Y_1), \nm(f_1(X_0), f_2(Y_1))) \approx_\eps \distr(U_m, \Copy(E_{01}, U_m)).
\]
For $\alpha_{01} < \eps$, the above distributions may be $1$-far; however, we can
still write the following for general $\alpha_{01} \in [0,1]$:
\begin{equation}
\label{eqn:nmreduction:E:01}
\alpha_{01} \distr(\nm(X_0, Y_1), \nm(f_1(X_0), f_2(Y_1))) \approx_\eps \alpha_{01} \distr(U_m, \Copy(E_{01}, U_m)),
\end{equation}
where in the above notation, we interpret distributions as vectors of probabilities
that can be multiplied by a scalar (i.e., $\alpha_{01}$) and use half the $\ell_1$
distance of vectors as the measure of proximity.
Similar results hold for
\[\cE_{10} := \nm(f_1(\cX_1), f_2(\cY_0))
= \nm(f_1(\cX_1), \cY_0)
\]
and
\[\cE_{11} := \nm(f_1(\cX_1), f_2(\cY_1)),\] so that for distributions
$\cD_{10}$ and $\cD_{01}$ over $\zo^m \cup \{ \same \}$ and independent random
variables $E_{10} \sim \cD_{10}$ and $E_{11} \sim \cD_{11}$,
\begin{equation}
\label{eqn:nmreduction:E:10}
\alpha_{10} \distr(\nm(X_1, Y_0), \nm(f_1(X_1), f_2(Y_0))) \approx_\eps \alpha_{10} \distr(U_m, \Copy(E_{10}, U_m)),
\end{equation}
and
\begin{equation}
\label{eqn:nmreduction:E:11}
\alpha_{11} \distr(\nm(X_1, Y_1), \nm(f_1(X_1), f_2(Y_1))) \approx_\eps \alpha_{11} \distr(U_m, \Copy(E_{11}, U_m)).
\end{equation}
We can also write, using the fact that $\nm$ is an ordinary extractor,
\begin{equation}
\label{eqn:nmreduction:E:00}
\alpha_{00} \distr(\nm(X_0, Y_0), \nm(f_1(X_0), f_2(Y_0))) \approx_\eps \alpha_{00} \distr(U_m, U_m).
\end{equation}

Denote by $\cD'_{01}$ the distribution $\cD_{01}$ conditioned on the complement of the event
$\{ \same \}$. Thus, $\cD'_{01}$ is a distribution over $\zo^m$. Similarly, define
$\cD'_{10}$ and $\cD'_{11}$ from $\cD_{10}$ and $\cD_{11}$ by conditioning on the event
$\zo^m \setminus \{ \same \}$. Observe that
\begin{equation}
\label{eqn:nmreduction:EE:01}
\distr(U_m, \Copy(E_{01}, U_m)) = p_{01} \distr(U_m, U_m) + (1-p_{01}) (\U_m, \cD'_{01}),
\end{equation}
where $p_{01} = \Pr[E_{01} = \same]$. Similarly, one can write
\begin{equation}
\label{eqn:nmreduction:EE:10}
\distr(U_m, \Copy(E_{10}, U_m)) = p_{10} \distr(U_m, U_m) + (1-p_{10}) (\U_m, \cD'_{10})
\end{equation}
and
\begin{equation}
\label{eqn:nmreduction:EE:11}
\distr(U_m, \Copy(E_{11}, U_m)) = p_{11} \distr(U_m, U_m) + (1-p_{11}) (\U_m, \cD'_{11}).
\end{equation}

Now, we can add up \eqref{eqn:nmreduction:E:01}, \eqref{eqn:nmreduction:E:10},
\eqref{eqn:nmreduction:E:11}, and \eqref{eqn:nmreduction:E:00}, using the triangle inequality,
and expand each right hand side according to \eqref{eqn:nmreduction:EE:01},
\eqref{eqn:nmreduction:E:10}, and \eqref{eqn:nmreduction:E:11} to deduce that
\begin{equation} \label{eqn:nmreduction:E:all}
\distr(\nm(X, Y), \nm(f_1(X), f_2(Y))) \approx_{4\eps}
p \distr(U_m, U_m) + (1-p) (\U_m, \cD')
\end{equation}
for some distribution $\cD'$ which is a convex combination 
\[ \cD' = \frac{1}{1-p}(\alpha_{01}(1-p_{01}) \cD'_{01} +
\alpha_{10}(1-p_{10}) \cD'_{10} + \alpha_{11}(1-p_{11})\cD'_{11})\]
 and coefficient $p = \alpha_{00} + \alpha_{01} p_{01}+
\alpha_{10} p_{10} + \alpha_{11} p_{11}$. 
Let $\cD$ be a distribution given by
\[
\cD := (1-p) \cD' + p \distr(\same),
\]
and observe that the right hand side of \eqref{eqn:nmreduction:E:all} is equal to
$\distr(U_m, \Copy(E, U_m))$, where $E \sim \cD$ is an independent random variable.
Thus, we conclude that
\begin{equation*} 
\distr(\nm(X, Y), \nm(f_1(X), f_2(Y))) \approx_{4\eps}
\distr(U_m, \Copy(E, U_m)),
\end{equation*}
which implies the non-malleability requirement of Definition~\ref{def:nmext:strict}.
\end{proof}

\subsection{From non-malleable extractors to non-malleable codes}
\label{sec:nmext:reduction}

In this section, we show a reduction from non-malleable extractors
to non-malleable codes. For concreteness, we focus on tampering functions
in the split-state model. That is, when the input is divided into 
two blocks of equal size, and the adversary may choose arbitrary functions
that independently tamper each block. It is straightforward to extend
the reduction to different families of tampering functions, for example:

\begin{enumerate}
\item When the adversary divides the input into $b \geq 2$ known parts, 
not necessarily of the same length, and applies an independent
tampering function on each block. In this case, a similar reduction from
non-malleable codes to multiple-source non-malleable extractors may be obtained.

\item When the adversary behaves as in the split-state model, but
the choice of the two parts is not known in advance. That is,
when the code must be simultaneously non-malleable for every
splitting of the input into two equal-sized parts. In this case,
the needed extractor is a non-malleable variation of the
\emph{mixed-sources extractors} studied by Raz and Yehudayoff \cite{ref:RY11}.

\end{enumerate}

We note that Theorem~\ref{thm:reduction} below (and similar theorems that
can be obtained for the other examples above)
only require non-malleable extraction from the uniform distribution. 
However, the reduction from arbitrary tampering functions to ones
without fixed points (e.g., Lemma~\ref{lem:nmext:nofixedpoint}) reduces
the entropy requirement of the source while imposing a structure
on the source distribution which is related to the family of tampering
functions being considered.

\begin{thm} \label{thm:reduction}
Let $\nm\colon \zo^n \times \zo^n \to \zo^k$ be a two-source non-malleable
$(n, n, \eps)$-extractor. Define a coding scheme $(\enc, \dec)$ with message length $k$ and block length
$2n$ as follows. The decoder $\dec$ is defined by $\dec(x) := \nm(x)$.

The encoder, given a message $s$, outputs a uniformly random string in $\nm^{-1}(s)$. 
Then, the pair $(\enc, \dec)$ is a non-malleable code with error
$\eps' := \eps (2^k+1)$ for the family of split-state adversaries.
\end{thm}

\begin{proof}
By construction, for every $s \in \zo^k$, $\dec(\enc(s)) = s$ with probability $1$.
It remains to verify non-malleability.

Take a uniformly random message $S \sim \U_k$, and let $Y := \enc(S)$ be its encoding.
First, we claim that $Y$ is close to be uniformly distributed on $\zo^{2n}$. 

\begin{claim}
The distribution of $\enc(S)$ is $\eps$-close to uniform.
\end{claim}

\begin{proof}
Let $Y' \sim \U_{2n}$, and $S' := \dec(Y') = \nm(Y')$. Observe that, since
$\nm$ is an ordinary extractor for the uniform distribution,
\begin{equation} \label{eqn:claim:SvsSS}
\distr(S') \approx_\eps \distr(S) = \U_k.
\end{equation}
On the other hand, since $\enc(s)$
samples a uniformly random element of $\nm^{-1}(s)$, it follows that
$\distr(\enc(S')) = \distr(Y') = \U_{2n}$. Since $S$ and $S'$ correspond
to statistically close distributions (by \eqref{eqn:claim:SvsSS}), this implies that
\[
\distr(\enc(S)) \approx_\eps \distr(\enc(S')) = \U_{2n}. \qedhere
\]
\end{proof}

In light of the above claim, in the sequel without loss of generality we can assume that $Y$ is
exactly uniformly distributed at the cost of an $\eps$ increase in the
final error parameter.

Let $Y = (Y_1, Y_2)$ where $Y_1, Y_2 \in \zo^n$.
The assumption that $\nm$ is a non-malleable extractor according to
Definition~\ref{def:nmext:strict} implies that
it is a non-malleable function with respect to the distribution of $Y$
and tampering function $f\colon \zo^{2n} \to \zo^{2n}$
\[
f(Y) := (f_1(Y_1), f_2(Y_2)),
\]
for any choice of the functions $f_1$ and $f_2$.
Let $\cD_f$ be the distribution $\cD$ defined in Definition~\ref{def:nmfunc}
that assures non-malleability of the extractor $\nm$,
and observe that its choice only depends on the functions $f_1$ and $f_2$ and not
the particular value of $S$. We claim that this is the right choice of $\cD_f$
required by Definition~\ref{def:nmCode}.

Let $S'' \sim \cD_f$ be sampled independently from $\cD_f$.
Since, by Definition~\ref{def:nmext:strict}, 
$\nm$ is a non-malleable function with respect to the distribution of
$Y$, Definition~\ref{def:nmfunc} implies that
\begin{equation*} 
\distr(\nm(Y), \nm(f(Y))) \approx_\eps \distr(\nm(Y), \Copy(S'', \nm(Y))),
\end{equation*}
which, after appropriate substitutions, simplifies to
\begin{equation} \label{eqn:nmExtDf}
\distr(S, \dec(f(\enc(S)))) \approx_\eps \distr(S, \Copy(S'', S)).
\end{equation}

Let $s \in \zo^k$ be any fixed message. We can now condition the above
equation on the event $S = s$, and deduce, using Proposition~\ref{prop:cond:general},
that
\begin{equation*} 
\distr(s, \dec(f(\enc(s)))) \approx_{\eps 2^k} \distr(s, \Copy(S'', s)),
\end{equation*}
or more simply, that
\begin{equation*} 
\distr(\dec(f(\enc(s)))) \approx_{\eps 2^k} \distr(\Copy(S'', s)),
\end{equation*}
which is the condition required to satisfy Definition~\ref{def:nmCode}.
It follows that $(\enc, \dec)$ is a non-malleable coding scheme
with the required parameters.
\end{proof}

We can now derive the following corollary, using the tools that we have developed so far.

\begin{coro} \label{coro:nmext:constant}
Let $\nm\colon \zo^n \times \zo^n \to \zo^m$ be a relaxed two-source non-malleable
$(k_1, k_2, \eps)$-extractor, where $m = \Omega(n)$, $n-k_1 = \Omega(n)$, 
$n-k_2 = \Omega(n)$, and $\eps = \exp(-\Omega(m))$. Then, there is a
$k = \Omega(n)$ such that the following holds.
Define a coding scheme $(\enc, \dec)$ with message length $k$ and block length
$2n$ (thus rate $\Omega(1)$) as follows. The decoder $\dec$, given $x \in \zo^{2n}$,
outputs the first $k$ bits of $\nm(x)$.
The encoder, given a message $x$, outputs a uniformly random string in $\dec^{-1}(x)$. 
Then, the pair $(\enc, \dec)$ is a non-malleable code with error
$\exp(-\Omega(n))$ for the family of split-state adversaries.
\end{coro}

\begin{proof}
Take $k = \frac{1}{2} \min\{ m, n-k_1, n-k_2, \log(1/\eps) \}$, which
implies that $k = \Omega(n)$ by the assumptions on parameters.  Furthermore,
we let $\eps' := 2^{-2k} \geq \eps$. 

Let $\nm'\colon \zo^n \times \zo^n \to \zo^k$ to be defined from 
$\nm$ by truncating the output to the first $k$ bits. 
Observe that as in ordinary extractors, truncating the output
of a non-malleable extractor does not affect any of the parameters other
than the output length. In particular, $\nm'$ is also
a relaxed two-source non-malleable $(k_1, k_2, \eps)$-extractor with
output length $\Omega(n)$.

In fact, our setup implies that $\nm'$ is a 
relaxed two-source non-malleable $(n - \log(1/\eps'), n - \log(1/\eps'), \eps')$-extractor with
output length $\Omega(n)$.
By Lemma~\ref{lem:nmext:nofixedpoint}, we see that $\nm'$ is a 
two-source non-malleable $(n, n, 4\eps')$-extractor.
We can now apply Theorem~\ref{thm:reduction} to conclude that
$(\enc, \dec)$ is a non-malleable code with error $4\eps'(2^k + 1) = \Omega(2^{-k})
= \exp(-\Omega(n))$ for split-state adversaries.
\end{proof}

\subsection{Existence bounds on non-malleable extractors}
\label{sec:nmext:existence}

So far we have introduced different notions of seedless non-malleable extractors
without focusing on their existence. In this section, we show that the
same technique used by \cite{ref:DW09} applies in a much more general setting and
can in fact show that non-malleable extractors exist with respect to every
family of randomness sources and every family of tampering adversaries, both of bounded size.
The main technical tool needed for proving this general claim is the following theorem.

\begin{thm} \label{thm:nmext:tool}
Let $\cX$ be a distribution over $\zo^n$ having min-entropy at least $k$,
and consider arbitrary functions $f\colon \zo^n \to \zo^n$ and $g\colon \zo^n \to \zo^d$.
Let $\nm\colon \zo^n \to \zo^m$ be a uniformly random function. Then, for any
$\eps > 0$, with probability at least $1-8\exp(2^{2m+d}-\eps^3 2^{k-6})$ the following hold.
\begin{enumerate}
\item The function $\nm$ extracts the randomness of $\cX$ even conditioned on the
knowledge of $g(X)$; i.e., 
\begin{equation} \label{eqn:thm:nmext:first}
\distr(g(X), \nm(X)) \approx_\eps \distr(g(X), \cU_m).
\end{equation}

\item Let $X \sim \cX$ and $U \sim \U_m$. Define the following random variable
over $\zo^m \cup \{ \same \}$:
\begin{equation} \label{eqn:nmext:Y}
Y := \begin{cases}
\same & \text{if $f(X) = X$} \\
\nm(f(X)) & \text{if $f(X) \neq X$}.
\end{cases}
\end{equation}
Then,
\begin{equation}
\label{eqn:thm:nmext:result}
\distr(g(X), \nm(X), \nm(f(X))) \approx_\eps \distr(g(X), U, \Copy(Y, U)).
\end{equation}

\item $\nm$ is a non-malleable function with respect to the distribution $\cX$
and tampering function $f$.
\end{enumerate}
\end{thm}

\begin{proof}
The proof borrows ideas from the existence proof of seeded non-malleable
extractors in \cite{ref:DW09}. The only difference is that we observe the same argument holds
in a much more general setting.

First, we observe that it suffices to prove \eqref{eqn:thm:nmext:result}, since
\eqref{eqn:thm:nmext:first} follows from \eqref{eqn:thm:nmext:result}.
Also, the result on non-malleability of the function $\nm$ follows from 
\eqref{eqn:thm:nmext:result}; in particular, one can use the explicit choice
\eqref{eqn:nmext:Y} of the random variable $Y$ in Definition~\ref{def:nmfunc}.
Thus, it suffices to prove \eqref{eqn:thm:nmext:result}.

Let $X \sim \cX$, $S := \supp(X)$, and $N := 2^n$, $K := 2^k$, $M := 2^m$,
$D := 2^d$. 
We will use the short-hands
\[
\nm_{g,f}(x) := (g(x), \nm(x), \nm(f(x))).
\]
and
\[
\nm_{g,f}(x,y) := (g(x), y, \nm(f(x))).
\]

Let $\beta = \Pr[f(X) \neq X]$, and let us first assume that $\beta \geq \eps/2$.
Let $\cX'$ be the distribution of $X$ conditioned on the event
$f(X) \neq X$, and $X' \sim \cX'$. The min-entropy of $\cX'$ is 
\[
H_\infty(\cX') \geq H_\infty(\cX) - \log(1/\beta) \geq k - \log(2/\eps).
\]
Instead of working with the tampering function $f$, for technical reasons it is more 
convenient to consider a related function $f'$
that does not have any fixed points.
Namely, let $f'\colon \zo^n \to \zo^n$ be any function such that
\[
\begin{cases}
f'(x) = f(x) & \text{if $f(x) \neq x$}, \\
f'(x) \neq f(x) & \text{if $f(x) = x$}.
\end{cases}
\]
By construction, $\Pr[f'(X) = X] = 0$.

Consider any distinguisher $h\colon \zo^d \times \zo^{2m} \to \zo$.
Let 
\[
P := \Pr_{X'}[h( \nm_{g,f'}(X') ) = 1]
\]
and
\[
\bar{P} := \Pr_{X'}[h( \nm_{g,f'}(X', U_m) ) = 1].
\]
Here, the probability is taken only over the random variable $X'$
and with respect to the particular realization of the function $\nm$.
That is, $P$ and $\bar{P}$ are random variables depending on the
randomness of the random function $\nm$.

For $x \in \zo^n$, we define
\[
P_x := h( \nm_{g,f'}(x) ),
\]
and
\[
\bar{P}_x := |\{ y \in \zo^m\colon h( \nm_{g,f'}(x, y) ) = 1 \}|/M.
\]
Again, $P_x$ and $\bar{P}_x$ are random variables depending only
on the randomness of the function $\nm$.
Since for any $x$, $\nm(x)$ and $\nm(f'(x))$ are uniformly distributed and
independent (due to the assumption that $f'(x) \neq x$), it follows
that $P_x$ and $\bar{P}_x$ have the same distribution as 
$h(g(x), \U_{2m})$ and thus
\[
\E[P_x - \bar{P}_x] = 0.
\]

As in \cite{ref:DW09}, we represent $f'$ as a directed graph $G=(V,E)$ with 
$V := \zo^n$ and $(x, y) \in E$ iff $f'(x) = y$. By construction,
$G$ has no self loops and the out-degree of each vertex is one. As shown
in \cite[Lemma~39]{ref:DW09}, $V$ can be partitioned as $V = V_1 \cup V_2$
such that $|V_1| = |V_2|$ and moreover, restrictions of $G$ to the vertices
in $V_1$ and $V_2$ (respectively denoted by $G_1$ and $G_2$) are both
acyclic graphs.

For $x \in \zo^n$, define $q(x) := \Pr[X' = x]$. It is clear that
\begin{equation*}
P = \sum_{x \in V} q(x) P_x,
\end{equation*}
and,
\begin{equation*}
\bar{P} = \sum_{x \in V} q(x) \bar{P}_x,
\end{equation*}
and consequently,
\begin{equation*}
P - \bar{P} = \sum_{x \in V} q(x) (P_x - \bar{P}_x) = \sum_{x \in V_1} q(x) (P_x - \bar{P}_x)
+ \sum_{x \in V_2} q(x) (P_x - \bar{P}_x).
\end{equation*}

Let $x_1, \ldots, x_{N/2}$ be the sequence of vertices of $G_1$
in reverse topological order. This means that for every $i \in [N/2-1]$,
\begin{equation}
f'(x_i) \notin \{ x_{i+1}, \ldots, x_{N/2} \}. \label{eqn:nmext:dag} 
\end{equation}

In general, the random variables $(P_x - \bar{P}_x)$ are not necessarily independent for
different values of $x$.  However, \eqref{eqn:nmext:dag} allows us to assert
conditional independence of these variables in the following form.
\begin{equation}
(\forall i \in [N/2-1])\colon \E[P_{x_{i+1}} - \bar{P}_{x_{i+1}} | P_1, \ldots, P_i,
\bar{P}_1, \ldots, \bar{P}_i] = 0.
\end{equation}
Therefore, the sequence
\[
\Big( \sum_{i=1}^j q(x_i) (P_{x_{i}} - \bar{P}_{x_{i}}) \Big)_{j \in [N/2]}
\]
forms a Martingale, and by Azuma's inequality, we have the concentration bound
\[
\Pr\Big[ \big|\sum_{x \in V_1} q(x) (P_x - \bar{P}_x)\big| > \eps/4 \Big] \leq 
2 \exp\Big(-\eps^2/\Big(32 \sum_{x \in V_1} q^2(x)\Big)\Big).
\]
The assumption on the min-entropy of $X'$, on the other hand, implies that
\[
\sum_{x \in V_1} q^2(x) \leq 2^{-k+\log(2/\eps)} \sum_{x \in V_1} q(x) \leq 2/(\eps K).
\]
A similar result can be proved for $V_2$; and using the above bounds
combined with triangle inequality we can conclude that
\[
\Pr[ |P - \bar{P}| > \eps/2 ] \leq 4 \exp(-\eps^3 K/64) =: \eta.
\]
That is, with probability at least $1-\eta$ over the randomness of $\nm$,
\[
\Big| \Pr_{X'}[h( \nm_{g,f'}(X') ) = 1] - \Pr_{X'}[h( \nm_{g,f'}(X', U_m) ) = 1] \Big|
\leq \eps/2.
\]
Since $f$ and $f'$ are designed to act identically on the support of $X'$, in the above result
we can replace $f'$ by $f$. Moreover,
by taking a union bound on all possible choices of the distinguisher,
we can ensure that with probability at least $1-\eta 2^{M^2 D}$, the realization
of $\nm$ is so that
\begin{equation*}
\distr( g(X'), \nm(X'), \nm(f(X')) ) \approx_{\eps/2} \distr( g(X'), U_m, \nm(f(X')) ).
\end{equation*}
We conclude that, regardless of the value of $\beta$, we can write
\begin{equation}
\label{eqn:nmext:approx}
\beta \distr( g(X'), \nm(X'), \nm(f(X')) ) \approx_{\eps/2} \beta \distr( g(X'), U_m, \nm(f(X')) ),
\end{equation}
where in the above notation, probability distributions are seen as vectors 
of probabilities that can be multiplied by a scalar $\beta$, and the distance
measure is half the $\ell_1$ distance between vectors (note that 
\eqref{eqn:nmext:approx} trivially holds for the case $\beta < \eps/2$).

Now we consider the distribution of $X$ conditioned on the event $f(X) = X$, that
we denote by $\cX''$. Again, we first assume that $1-\beta \geq \eps/2$, in which
case we get
\[
H_\infty(\cX'') \geq k - \log(2/\eps).
\]
If so, a similar argument as above (in fact, one that does not require the use
of partitioning of $G$ and using Maringale bounds since the involved random
variables are already independent) shows that with probability 
at least $1-\eta 2^{M^2 D}$ over the choice of $\nm$, for $U \sim \U_m$ we have
\begin{equation*}
\distr( g(X''), \nm(X''), \nm(f(X'')) ) \approx_{\eps/2} \distr( g(X''), U, U ).
\end{equation*}
For general $\beta$, we can thus write
\begin{equation}
\label{eqn:nmext:approx:b}
(1-\beta) \distr( g(X''), \nm(X''), \nm(f(X'')) ) \approx_{\eps/2} (1-\beta) \distr( g(X''), U, U ).
\end{equation}
Now, we may add up \eqref{eqn:nmext:approx:b} and \eqref{eqn:nmext:approx:b} and
use the triangle inequality to deduce that, with probability 
at least $1-2\eta 2^{M^2 D}$ over the choice of $\nm$, 
\begin{equation}
\label{eqn:nmext:approx:c}
\distr( g(X), \nm(X), \nm(f(X)) ) \approx_{\eps} 
\beta \distr( g(X'), U, \nm(f(X')) ) + (1-\beta) \distr( g(X''), U, U ).
\end{equation}
The result \eqref{eqn:thm:nmext:result} now follows after observing that
the convex combination on the right hand side of \eqref{eqn:nmext:approx:c}
is the same as $\distr(g(X), U, \Copy(Y, U))$.
\end{proof}

As mentioned before, the above theorem is powerful enough to show existence of any desired form
of non-malleable extractors, as long as the class of sources and the family 
of tampering functions (which are even allowed to have fixed points) are
of bounded size. In particular, it is possible to use the theorem to recover the result in
\cite{ref:DW09} on the existence of strong seeded non-malleable extractors
by considering both the seed and input of the extractor as an $n$-bit string,
and letting ``the side information function'' $g(X)$ be one that simply outputs the 
seed part of the input. The family of tampering functions, on the other hand, would be all functions
that act on the portion of the $n$-bit string corresponding to the 
extractor's seed.

For our particular application, we apply Theorem~\ref{thm:nmext:tool}
to show existence of two-source non-malleable extractors. In fact, it is
possible to prove existence of \emph{strong} two-source extractors in the
sense that we may allow any of the two sources revealed to the distinguisher,
and still guarantee extraction and non-malleability properties. However,
such strong extractors are not needed for our particular application.

\begin{thm}
\label{thm:nmext:existence:twosrc}
Let $\nm\colon \zo^n \times \zo^n \to \zo^m$ be a uniformly random function.
For any $\gamma, \eps > 0$ and parameters $k_1, k_2 \leq n$, with probability
at least $1-\gamma$ the function $\nm$ is a two-source non-malleable 
$(k_1, k_2, \eps)$-extractor provided that
\begin{gather*}
2m \leq k_1 + k_2 - 3\log(1/\eps) - \log\log(1/\gamma), \\
\min \{k_1, k_2\} \geq \log n + \log\log(1/\gamma) + O(1).
\end{gather*}
\end{thm}

\begin{proof}
First we note that, similar to ordinary extractors, Definition~\ref{def:nmext:strict} 
remains unaffected if one only considers random sources
where each component is a flat distribution.

Let $K_1 := 2^{k_1}$, $K_2 := 2^{k_2}$, $N := 2^n$, $M := 2^m$.
Without loss of generality, assume that $K_1$ and $K_2$ are integers.
Let $\mathfrak{X}$ be the class of distributions $\cX = (\cX_1, \cX_2)$ over
$\zo^n \times \zo^n$ such that $\cX_1$ and $\cX_2$ are flat sources
with min-entropy at least $k_1$ and $k_2$, respectively. Note that the
min-entropy of $\cX$ is at least $k_1 + k_2$. Without loss
of generality, we assume that $k_1 \leq k_2$. The number
of such sources can be bounded as
\[
|\mathfrak{X}| \leq \binom{N}{K_1} \binom{N}{K_2} \leq N^{K_1+K_2}
\leq N^{2 K_2}.
\]
The family $\cF$ of tampering functions can be written as $\cF = \cF_1 \times \cF_2$, 
where $\cF_1$ and $\cF_2$ contain functions that act on the first and second $n$ bits,
respectively. For the family $\cF_1$, it suffices to only consider functions that
act arbitrarily on some set of $K_1$ points in $\zo^n$, but are equal to the identity
function on the remaining inputs. This is because a tampering function $f_1 \in \cF_1$
will be applied to some distribution $\cX_1$ which is only supported on a
particular set of $K_1$ points in $\zo^n$, and thus the extractor's behavior on $\cX_1$ is
not affected by how $f_1$ is defined outside the support of $\cX_1$.
From this observation, we can bound the size of $\cF$ as
\[
|\cF| \leq \binom{N}{K_1} N^{K_1} \cdot \binom{N}{K_2} N^{K_2} \leq N^{2(K_1+K_2)} \leq
N^{4K_2}.
\]
Now, we can apply Theorem~\ref{thm:nmext:tool} on the input domain
$\zo^n \times \zo^n$. The choice of the function $g$ is not important
for our result, since we do not require two-source extractors that are
strong with respect to either of the two sources. We can thus set
$g(x) = 0$ for all $x \in \zo^{2n}$. By taking a union bound on all choices
of $\cX \in \mathfrak{X}$ and $(f_1, f_2) \in \cF$, we deduce that
the probability that $\nm$ fails to satisfy Definition~\ref{def:nmext:strict} for some
choice of the two sources in $\mathfrak{X}$ and tampering function in $\cF$ is
at most
\[8 \exp(2 M^2-\eps^3 K_1 K_2/16) |\mathfrak{X}| \cdot |\cF| \leq
8 N^{4K_2} \exp(2 M^2- \eps^3 K_1 K_2/64).
\] 
This probability can be made less than $\gamma$ provided that
\begin{gather*}
2m \leq k_1 + k_2 - 3\log(1/\eps) - \log\log(1/\gamma), \\
k_1 \geq \log n + \log\log(1/\gamma) + O(1),
\end{gather*}
as desired.
\end{proof}

We are finally ready to prove that there are non-malleable two-source extractors 
defining coding schemes secure in the split-state model 
and achieving constant rates; in particular, arbitrarily close to $1/5$.

\begin{coro} \label{coro:split:main}
For every $\alpha > 0$, there is a choice of $\nm$ in Theorem~\ref{thm:reduction}
that makes $(\enc, \dec)$ a non-malleable coding scheme against split-state
adversaries achieving rate $1/5-\alpha$ and error $\exp(-\Omega(\alpha n))$.
\end{coro}

\begin{proof}
First, for some $\alpha'$, we use Theorem~\ref{thm:nmext:existence:twosrc} to show that
if $\nm\colon \zo^n \times \zo^n \to \zo^k$ is randomly chosen, 
with probability at least $.99$ it is a two-source non-malleable
$(n, n, 2^{-k(1+\alpha')})$-extractor, provided that
\[
k \leq n - (3/2) \log(1/\eps) - O(1) = n - (3/2) k(1+\alpha') - O(1),
\]
which can be satisfied for some $k \geq (2/5) n - \Omega(\alpha' n)$.
Now, we can choose $\alpha' = \Omega(\alpha)$ so as to ensure that
$k \geq 2n(1-\alpha)$ (thus, keeping the rate above $1-\alpha$) while having
$\eps \leq 2^{-k} \exp(-\Omega(\alpha n))$. We can now apply
Theorem~\ref{thm:reduction} to attain the desired result.
\end{proof}

\bibliographystyle{plain}
\bibliography{\jobname} 

\appendix

\section{Construction of LECSS codes} \label{sec:lecss}

In this section, we recall a well-known construction of LECSS
codes based on linear error-correcting codes \cite{ref:nmc, ref:CCGHV07}. 
Construction~\ref{constr:lecss} defines the reduction.

\begin{constr} 
  \caption{Explicit construction of LECSS codes from linear codes.}

  \begin{itemize}
  \item {\it Given: } A $k \times n$ matrix $G$ over $\F_q$, where
$q$ is a power of two and $n \geq k$ such that
\begin{enumerate}
\item Rows of $G$ span a code with relative
distance at least $\delta > 0$,
\item For some $k_0 \in [k]$, the first $k_0$ rows of $G$
span a code with dual relative distance at least $\tau > 0$.
\end{enumerate}

  \item {\it Output: } A coding scheme $(\enc, \dec)$ 
of block length 
$N := n \log q$ and message length $K := (k-k_0) \log q$.

\item {\it Construction of the encoder $\enc(s)$, given a message
$s \in \zo^K$: }

\begin{enumerate}
\item Pick a uniformly random vector $(s_1, \ldots, s_{k_0}) \in \F_q^{k_0}$.

\item Interpret $s$ as a vector over $\F_q$; namely,
$(s_{k_0+1}, \ldots, s_{k}) \in \F_q$.

\item Output $(s_1, \ldots, s_k) \cdot G \in \F_q^{n}$ in
binary form (i.e., as a vector in $\zo^{N}$).
\end{enumerate}

\item {\it Construction of the decoder $\dec(w)$, given an 
input $w \in \zo^N$: }

\begin{enumerate}
\item Interpret $w$ as a vector $(w_1, \ldots, w_n) \in \F_q^n$.

\item If there is a vector $(s_1, \ldots, s_k) \in \F_q^k$
such that $(s_1, \ldots, s_k) \cdot G = (w_1, \ldots, w_n)$,
output $(s_{k_0+1}, \ldots, s_{k}) \in \F_q^{k-k_0}$
in binary form (i.e., as a vector in $\zo^K$).
Otherwise, output $\perp$.
\end{enumerate}
  \end{itemize}
  \label{constr:lecss}
\end{constr}

The main tool that we use is the following lemma, which appears (in a slightly different form)
in \cite{ref:nmc} (which in turn is based on \cite{ref:CCGHV07}). We include a proof for completeness.

\begin{lem} \label{lem:lecss}
The pair $(\enc, \dec)$ of Construction~\ref{constr:lecss}
is a $(\delta N/\log q, \tau N/\log q)$-linear error-correcting
coding scheme.
\end{lem}

\begin{proof}
First, observe that the linearity condition of Definition~\ref{def:lecss} follows
from the fact that $\enc$ is an injective linear function of $(s_1, \ldots, s_k)$
as defined in Construction~\ref{constr:lecss}. Furthermore, the distance property of the
coding scheme follows from the fact that $\enc$ encodes
an error-correcting of distance at least $\delta n = \delta N/(\log q)$.

In order to see the bounded independence property of Definition~\ref{def:lecss},
consider a fixed message $s \in \zo^K$, which in turn fixes the vector
$(s_{k_0+1}, \ldots, s_k)$ in Construction~\ref{constr:lecss}. Let $G_0$
denote the sub-matrix of $G$ defined by the first $k_0$ rows.
Consider the vector $S' \in \F_q^n$ given by
\[
S' := (s_1, \ldots, s_k) \cdot G = 
(s_1, \ldots, s_{k_0}) \cdot G_0 + a,
\]
where $a \in \F_q^n$ is an affine shift uniquely determined by $s$.
Recall that the assumption on the dual distance of the code spanned by the rows of $G_0$
implies that every $\tau n$ columns of $G_0$ are linearly independent.
Since $(s_1, \ldots, s_{k_0})$ is a uniformly random vector, this
implies that the restriction of $S'$ to any set of $\tau n = \tau N/(\log q)$ 
coordinates is uniformly random (as a vector in $\F_q^{\tau n}$).
Since $\enc(s)$ is the bit-representation of $S'$, it follows that
the random vector $\enc(s)$ is $(\tau N/(\log q))$-wise independent.
\end{proof}

\subsection*{Instantiation using Reed-Solomon codes}

A simple way to instantiate Construction~\ref{constr:lecss} is
using Reed-Solomon codes. For a target rate parameter
$r := 1-\alpha$, we set up the parameters as follows.
For simplicity, assume that $n$ is a power of two.
\begin{enumerate}
\item The field size is $q := n$. Therefore, $N = n \log n$.

\item Set $k := \lceil n(1-\alpha/2) \rceil$ and
$k_0 := \lfloor \alpha n/2 \rfloor$. Therefore,
$K := (k-k_0) \log q \geq n(1-\alpha) \log n$,
which ensures that the rate of the coding scheme
is at least $1-\alpha$.

\item Since $G$ generates a Reed-Solomon code, which is
an MDS code, we have $\delta = 1-k/n \geq \alpha/2-1/n
= \Omega(\alpha)$.

\item We note that the matrix $G$ is a $k \times n$ Vandermonde matrix
whose first $k_0$ rows also form a Vandermonde matrix spanning
a Reed-Solomon code. The dual distance of the code formed
by the span of the first $k_0$ rows of $G$ is thus equal to
$\tau = k_0/n \geq \alpha/2 - 1/n = \Omega(\alpha)$.
\end{enumerate}

In particular, Lemma~\ref{lem:lecss} applied to the above
set up of the parameters implies that the resulting coding scheme
is an $(\Omega(\alpha N/\log n), \Omega(\alpha N/\log n))$-linear
error-correcting secret sharing code.

When $n$ is not a power of two, it is still possible to pick
the least $q \geq n$ which is a power of two and obtain
similar results. In general, we have the following corollary
of Lemma~\ref{lem:lecss}.

\begin{coro} \label{coro:lecss}
For every integer $n \geq 1$ and $\alpha \in (0,1)$, 
there is an explicit construction
of a binary coding scheme $(\enc, \dec)$ of block length $n$ and
message length $k \geq n(1-\alpha)$ which is an
$(\Omega(\alpha n/\log n), \Omega(\alpha n/\log n))$-linear
error-correcting secret sharing code. \qed
\end{coro}

\section{Useful tools}

In some occasions in the paper, we deal with a chain of correlated random variables
$0 = X_0, X_1, \ldots, X_n$ where we wish to understand an event depending on $X_i$
conditioned on the knowledge of the previous variables. That is, we wish to understand
\[
\E[f(X_i) | X_0, \ldots, X_{i-1}].
\]
The following proposition shows that in order to understand the above quantity,
it suffices to have an estimate with respect to a more restricted event than
the knowledge of $X_0, \ldots, X_{i-1}$. Formally, we can state the following,
where $X$ stands for $X_i$ in the above example and $Y$ stands for $(X_0, \ldots, X_{i-1})$.

\begin{prop} \label{prop:restriction}
Let $X$ and $Y$ be possibly correlated random variables and let
$Z$ be a random variable such that the knowledge of $Z$ determines $Y$; 
that is, $Y = f(Z)$ for some function $f$. Suppose that
for every possible outcome of the random variable $Z$, namely, for
every $z \in \supp(Z)$, and for some real-valued function $g$, we have 
\begin{equation} \label{eqn:prop:restriction}
\E[g(X) | Z = z] \in I.
\end{equation}
for a particular interval $I$. Then, for every $y \in \supp(Y)$,
\[
\E[g(X) | Y = y] \in I.
\]
\end{prop}
\begin{proof}
Let $T = \{ z \in \supp(Z)\colon f(z) = y \}$, and let 
$p(z) := \Pr[Z = z | Y = y]$. Then,
\[
\E[g(X) | Y = y] = \sum_{z \in T} p(z) \E[g(X) | Z = z].
\] 
Since by \eqref{eqn:prop:restriction}, each $\E[g(X) | Z = z]$ lies in $I$ 
and $\sum_{z \in T} p(z) = 1$, we deduce that
\[
\E[g(X) | Y = y] \in I.
\]
\end{proof}

\begin{prop} \label{prop:uniformity}
Let the random variable $X \in \zo^n$ be uniform on a set of size
at least $(1-\eps)2^n$. Then, $\cD(X)$ is $(\eps/(1-\eps))$-close to $\U_n$.
\end{prop}

\begin{prop} \label{prop:cond:general}
Let $\cD$ and $\cD'$ be distributions over the same finite space $\Omega$,
and suppose they are $\eps$-close to each other. Let
$E \subseteq \Omega$ be any event such that $\cD(E) = p$.
Then, the conditional distributions $\cD|E$ and $\cD'|E$ are
$(\eps/p)$-close.
\end{prop}

\begin{lem} \label{lem:nmc:joint}
Let $(\enc, \dec)$ be a coding scheme of message length $k$ which is non-malleable with respect to 
a family $\cF$ of adversaries with error $\eps$. Let $S \in \zo^k$ be a message drawn
randomly according to any distribution and $S' := \dec(f(\enc(S)))$ for 
some $f \in \cF$. Then, there is an independent random variable $S'' \in \zo^k$ and
parameter $\alpha \in [0,1]$ only depending on the code and $f$ such that
\[
\distr(S, S') \approx_\eps \alpha \cdot \distr(S, S) + (1-\alpha) \cdot \distr(S, S''). 
\]
\end{lem}

\begin{proof}
Let $\cD_f$ be the distribution from Definition~\ref{def:nmCode} over $\zo^k \cup \{\same\}$
and let $\alpha = \cD_f(\same)$ be the probability assigned to $\same$ by $\cD_f$.
Let $S_0 \sim \cD_f$ be an independent random variable and $S''$ be an
independent random variable drawn from the distribution of $S_0$ conditioned
on the event $S_0 \neq \same$. By Definition~\ref{def:nmCode},
we have 
\[
\distr(S, S') \approx_\eps \distr(S, \Copy(S_0, S)),
\]
which can be seen by applying the definition for every fixing of $S$ and taking a
convex combination. In turn, we have
\[
\distr(S, \Copy(S_0, S))  = \alpha  \distr(S, S) + (1-\alpha) \distr(S, S'').
\]
which completes the proof.
\end{proof}

\begin{prop} \label{prop:pBiased}
Let $\cD$ be the distribution of $n$ independent bits, where each
bit is $\eps$-close to uniform. Then, $\cD$ is $O(n\eps)$-close to $\cU_n$.
\end{prop}

\begin{proof}
Let $x \in \zo^n$ be any fixed string. Then
\[
\cD(x) \leq (1/2+\eps)^n = 2^{-n} (1+2\eps)^n \leq 2^{-n} (1+O(\eps n)).
\]
Similarly, one can show that $\cD(x) \geq 2^{-n} (1-O(\eps n))$.
Now, the claim follows from the definition of statistical distance
and using the above bounds for each $x$.
\end{proof}

We will use the following tail bound on summation of possibly dependent 
random variables, which is a direct consequence of Azuma's inequality.

\begin{prop} \label{prop:simpleAzuma}
Let $0 = X_0, X_1, \ldots, X_n$ be possibly correlated indicator random variables
such that for every $i \in [n]$ and for some $\gamma \geq 0$,
\[
\E[X_i | X_0, \ldots, X_{i-1}] \leq \gamma.
\]
Then, for every $c \geq 1$,
\[
\Pr[\sum_{i=1}^n X_i \geq cn\gamma] \leq \exp(- n\gamma^2 (c-1)^2/2),
\]
or equivalently, for every $\delta > \gamma$,
\[
\Pr[\sum_{i=1}^n X_i \geq n\delta] \leq \exp(- n(\delta-\gamma)^2/2).
\]
\end{prop}

\begin{proof} See \cite{ref:CG1} for a proof.
\end{proof}

In a similar fashion (using Azuma's inequality for sub-martingales
rather than super-martingales in the proof), we may obtain a tail bound
when we have a lower bound on conditional expectations.

\begin{prop} \label{prop:simpleAzuma:sub}
Let $0 = X_0, X_1, \ldots, X_n$ be possibly correlated random variables in $[0,1]$
such that for every $i \in [n]$ and for some $\gamma \geq 0$,
\[
\E[X_i | X_0, \ldots, X_{i-1}] \geq \gamma.
\]
Then, for every $\delta < \gamma$,
\[
\Pr[\sum_{i=1}^n X_i \leq n\delta] \leq \exp(- n(\delta-\gamma)^2/2). \qedhere \]
\end{prop}

The lemma below shows that it is possible to sharply approximate a distribution
$\cD$ with finite support by sampling possibly correlated random variables
$X_1, \ldots, X_n$ where the distribution of each $X_i$ is close to $\cD$
conditioned on the previous outcomes, and computing the empirical distribution
of the drawn samples.

\begin{lem}\cite{ref:CG1} \label{lem:distrLearning:dependent}
Let $\cD$ be a distribution over a finite set $\Sigma$ such that
$|\supp(\cD)| \leq r$. For any
$\eta, \eps, \gamma > 0$ such that $\gamma < \eps$, there is a choice of 
\[
n_0 = O((r+2+\log(1/\eta))/(\eps-\gamma)^2)
\] such that for every $n \geq n_0$ the following holds. 
Suppose
$0 = X_0, X_1, \ldots, X_n \in \Sigma$ are possibly correlated random variables such that
for all $i \in [n]$ and all values $0 = x_0, x_1 \ldots, x_n \in \supp(\cD)$,
\[
\distr(X_i | X_0 = x_0, \ldots, X_{i-1} = x_{i-1}) \approx_\gamma \cD.
\]
Then, with probability at least $1-\eta$,
the empirical distribution of the outcomes 
$X_1, \ldots, X_n$ is $\eps$-close to $\cD$.
\end{lem}

\end{document}